\documentclass[thm-restate,runningheads,anonymous]{llncs}

\usepackage[T1]{fontenc}
\usepackage{graphicx}
\usepackage{amsmath, amssymb, amsfonts}

\usepackage{todonotes}
\usepackage{xspace}

\mathchardef\mh="2D

\usepackage{tcolorbox}
\usepackage{xcolor}
\definecolor{grey}{RGB}{245,245,245}

\usepackage{enumitem}
\colorlet{mix}{red!50!black}

\usepackage{hyperref}
\hypersetup{colorlinks={true},linkcolor={blue},citecolor=blue}

\usepackage{cleveref}

\newcommand{\markone}{$\mathtt{Mark\mh 1}$}
\newcommand{\markk}{\mathtt{Mark}}

\newcommand{\card}[1]{\ensuremath{{\vert {#1} \vert }}} %%% cardinality of a set
\newcommand{\set}[1]{\ensuremath{\left\{ {#1} \right\}}} %%% set notation
\newcommand{\phcags}{{proper Helly circular-arc graphs}\xspace}
\newcommand{\phcag}{{proper Helly circular-arc graph}\xspace}
\newcommand{\pac}{{\sc Proper Helly Circular-arc Vertex Deletion}\xspace}
\newcommand{\phcad}{{\sc PHCAVD}\xspace}
\newcommand{\obs}{{\mathbb{O}}\xspace}
\newcommand{\prv}{\mathsf{prev}}
\newcommand{\nxt}{\mathsf{next}}
\newcommand{\glb}{\mathsf{g}}
\newcommand{\OO}{\mathcal{O}}
\newcommand{\cO}{\mathcal{O}}

\newcommand{\C}{\mathcal}

\newcommand{\dist}{\text{dist}}

\newcommand{\redundC}{{5}\xspace}
\newcommand{\calW}{{\mathcal {W}}\xspace}

\newcommand{\copyg}{{\mathsf {copy}}\xspace}

\newcommand{\cl}{{\mathcal{Q}_i}\xspace}
\newcommand{\cd}{{\mathbb{Q}}\xspace}
\newcommand{\cdc}{{\mathbb{Q}_\mathcal{C}}\xspace}

\newcommand{\monad}{\texttt{Monad}\xspace}
\newcommand{\centre}{\texttt{centre}\xspace}
\newcommand{\hole}{\texttt{M-Hole}\xspace}
\newcommand{\mhole}{\texttt{M-Hole}\xspace}

\newcommand{\NP} {{\sf NP}}
\newcommand{\FPT} {{\sf FPT}}

\newcommand{\NPC} {{\sf NP}-complete\xspace}

\newcommand{\no}{\texttt{No}\xspace} 
\newcommand{\yes}{\texttt{Yes}\xspace} 
\newcommand{\appIVD}{\texttt{ApproxPHCAD}\xspace}
\newcommand{\redIVD}{\texttt{RedundantPHCAD}\xspace}

\newcommand{\cyclebound}{\ensuremath{12}}

\newcommand{\defparproblem}[4]{
	\vspace{1mm}
	\noindent\fbox{
		\begin{minipage}{0.96\textwidth}
			\begin{tabular*}{\textwidth}{@{\extracolsep{\fill}}lr} #1  \\ \end{tabular*}
			{\bf{Input:}} #2  \\
			{\bf{Parameter:}} #3\\
			{\bf{Output:}} #4
		\end{minipage}
	}
	\vspace{1mm}
}

\newtheorem{reduction rule}{Reduction Rule}
\newtheorem{marking scheme}{Marking Scheme}
\newtheorem{observation}{Observation}
\newtheorem{clm}{Claim}

\begin{document}
\title{A Polynomial Kernel for Proper Helly Circular-arc Vertex Deletion}
\titlerunning{A Polynomial Kernel for Proper Helly Circular-arc Vertex Deletion}

% \author{Akanksha Agrawal}{Indian Institute of Technology Madras, India}{akanksha@cse.iitm.ac.in}{https://orcid.org/0000-0002-0656-7572}{Supported by New Faculty Initiation Grant no. NFIG008972}

% \author{Satyabrata Jana}{The Institute of Mathematical Sciences, HBNI, Chennai, India}{satyamtma@gmail.com}{https://orcid.org/0000-0002-7046-0091}{}

% \author{Abhishek Sahu}{National Institute of Science, Education and Research, An OCC of Homi Bhabha National Institute, Bhubaneswar 752050, Odisha, India.}{asahuiitkgp@gmail.com}{}{}
\author{Akanksha Agrawal\inst{1} \and
Satyabrata Jana\inst{2} \and
Abhishek Sahu\inst{3}}
\authorrunning{A. Agrawal, S. Jana, A. Sahu}
%% First names are abbreviated in the running head.
%% If there are more than two authors, 'et al.' is used.
%%
\institute{Indian Institute of Technology Madras, India \and
The Institute of Mathematical Sciences, HBNI, Chennai, India \and National Institute of Science, Education and Research, An OCC of Homi Bhabha National Institute, Bhubaneswar 752050, Odisha, India
\\
\email{akanksha@cse.iitm.ac.in, \{satyamtma,asahuiitkgp\}@gmail.com}}
\maketitle              % typeset the header of the contribution
\begin{abstract}
	%!TEX root = main.tex

A {\em proper Helly circular-arc graph} is an intersection graph of a set of arcs on a circle such that none of the arcs properly contains any other arc and every set of pairwise intersecting arcs has a common intersection. The \pac\ problem takes as input a graph $G$ and an integer $k$, and the goal is to check if we can remove at most $k$ vertices from the graph to obtain a proper Helly circular-arc graph; the parameter is $k$. Recently, Cao et al.~[MFCS 2023] obtained an FPT algorithm for this (and related) problem. In this work, we obtain a polynomial kernel for the problem.  

\end{abstract}

	\section{Introduction}

The development of parameterized complexity is much owes much to the study of graph modification problems, which have inspired the evolution of many important tools and techniques. One area of parameterized complexity is data reduction, also known as kernelization, which focuses on the family of graphs $\C{F}$ and the {\sc $\C{F}$-Modification} problem. Given a graph $G$ and an integer $k$, this problem asks whether it is possible to obtain a graph in $\C{F}$ using at most $k$ modifications in $G$, where the modifications are limited to vertex deletions, edge deletions, edge additions, and edge contractions. The problem has been extensively studied, even when only a few of these operations are allowed.
\smallskip

Here we deal on the parameterization of the {\sc $\C{F}$-Vertex Deletion} problem, which is a special case of {\sc $\C{F}$-Modification} where the objective is to find the minimum number of vertex deletions required to obtain a graph in $\C{F}$. This problem encompasses several well-known \NP-complete problems, such as {\sc Vertex Cover}, {\sc Feedback Vertex Set}, {\sc Odd Cycle Transversal}, {\sc Planar Vertex Deletion}, {\sc Chordal Vertex Deletion}, and {\sc Interval Vertex Deletion}, which correspond to $\C{F}$ being the family of graphs that are edgeless, forests, bipartite, planar, chordal and interval, respectively. Unfortunately, most of these problems are known to be {\NPC}, and therefore have been extensively studied in paradigms such as parameterized complexity designed to cope with \NP-hardness. There have been many studies on this topic, including those referenced in this paper, but this list is not exhaustive.
\smallskip

In this article, we focus on the {\sc $\C{F}$-Vertex Deletion} problem, specifically when $\C{F}$ refers to the family of proper Helly circular-arc graphs. We refer to this problem as \pac(\phcad) for brevity. A circular-arc graph is a graph whose vertices can be assigned to arcs on a circle such that there is an edge between two vertices if and only if their corresponding arcs intersect. If none of the arcs properly contains one another, the graph is a proper circular-arc graph. These graphs have been extensively studied, and their structures and recognition are well understood \cite{duran2014structural,kaplan2009certifying,lin2009characterizations}. These graphs also arise naturally when considering the clique graphs of a circular-arc graph. However, the lack of the Helly property, which dictates that every set of intersecting arcs has a common intersection, contributes to the complicated structures of circular-arc graphs. A Helly circular-arc graph is a graph that admits a Helly arc representation. All interval graphs are Helly circular-arc graphs since every interval representation is Helly. The class of proper Helly circular-arc graphs lies between proper circular-arc graphs and proper interval graphs. A graph is a proper Helly circular-arc graph if it has a proper and Helly arc representation. Circular-arc graphs are a well-studied graph class due to their intriguing combinatorial properties and modeling power \cite{Golumbic80}.  Additionally, there exists a linear-time algorithm to determine if a given graph is a circular-arc graph and construct a corresponding arc representation if so \cite{mcconnell2003linear}, even for Helly circular-arc graphs, such algorithm exists \cite{DBLP:conf/cocoon/LinS06}.\\
For graph modification problems, the number of allowed modifications, $k$, is considered the {\em parameter}. With respect to $k$, such a problem is said to be {\em fixed-parameter tractable} (\FPT)  if it admits an algorithm running in time $f(k)n^{\OO(1)}$ for some computable function $f$. Also, the problem is said to have a polynomial kernel if in polynomial time (with respect to the size of the instance) one can obtain an equivalent instance of polynomial size (with respect to the parameter), i.e., for any given instance $(G,k)$ of the problem, it can be reduced in time $n^{\cO(1)}$ to an equivalent instance $(G',k')$ where
$|V(G')|$ and $k'$ are upper bounded by $k^{\cO(1)}$. A kernel for a problem immediately implies that it admits an \FPT\ algorithm, but kernels are also interesting in their own right. In particular, kernels allow us to model the performance of polynomial-time preprocessing algorithms. The field of kernelization has received considerable attention, especially after the introduction of methods to prove kernelization lower bounds~\cite{BDFH09}. We refer to the books~\cite{paramalgoCFKLMPPS,DowneyFbook13}, for a detailed treatment of the area of kernelization.\\
Designing polynomial kernels for problems such as {\sc Chordal Vertex Deletion} \cite{AgrawalLMSZ17} and {\sc Interval Vertex Deletion} \cite{AgrawalM0Z19} posed several challenges. In fact, kernels for these problems were obtained only recently, after their status being open for quite some time. \pac has remained an interesting problem in this area.  Recently, Cao et al.~\cite{DBLP:journals/corr/abs-2202-00854} studied this problem and showed that it admits a factor $6$-approximation algorithm, as well as an \FPT algorithm that runs in time $6^k \cdot n^{\OO(1)}$. \\
  A natural follow-up question to the prior work on this problem is to check whether \phcad admits a polynomial
kernel. In this paper, we resolve this question in the
affirmative way.

%These tools can be used to show that a polynomial kernel~\cite{BDFH09,D15,FS11,HKSWW15}, or a kernel of a specific size~\cite{DM12,DM14,HW12} for concrete problems would imply an unlikely complexity-theoretic collapse. 
%%We refer the reader to the survey articles by Kratsch~\cite{Kratsch14} and
%%Lokshtanov et al.~\cite{LokshtanovMS12} for recent developments. 

\smallskip

\defparproblem{\pac(\phcad)}{ A graph $G$ and an integer $k$.}{$k$}{Does there exist a subset $ S \subseteq V(G) $ of size at most $ k $ such that $ G-S $ is a \phcag?}
\medskip

  %In this paper, we present a polynomial kernel for this problem.

\begin{theorem}\label{theo:poly_kernel}
	\pac admits a polynomial kernel.
\end{theorem}

\subsection{Methods}
Our kernelization heavily uses the characterization of \phcags in terms of their {\em forbidden induced subgraphs}, also called {\em obstructions}. Specifically, a graph $H$ is an obstruction to the class of \phcags if $H$ is not \phcag but $H - \{v\}$ is \phcag for every vertex $v \in V(H)$. A graph $G$ is a \phcag if and only if it does not contain any of the following obstructions as induced subgraphs, which are $ \overline{C_3^*}$ (claw),
$S_3$ (tent),  $\overline{S_3}$ (net),$W_4$ (wheel of size 4) , $W_5$ (wheel of size 5), $\overline{C_6}$ as well as a family of graphs: $C_{\ell}^*,~\ell \geq 4$ referred to as a $\monad$ of size $\ell$ (see \cref{fig:obs})~\cite{DBLP:journals/corr/abs-2202-00854,DBLP:journals/dam/LinSS13}. We call any obstruction of size less than $\cyclebound$ a {\em small obstruction}, and call all other obstructions \emph{large obstructions}. Note that every \emph{large obstruction} is a \monad  of size at least $\cyclebound$. 
\smallskip

The first ingredient of our kernelization algorithm is the factor $6$ polynomial-time approximation algorithm for \phcad given by Cao et al.~\cite{DBLP:journals/corr/abs-2202-00854}. We use this algorithm to obtain an approximate solution of size at most $6k$, or conclude that there is no solution of size at most $k$.  We grow (extend) this approximate solution to a set $T_1$ of size $\OO(k^{\cyclebound})$, such that 
%for any solution of size at most $k$, we have that 
every set $Y \subseteq V (G)$ of size at most $k$ is a minimal hitting set for all \emph{small} obstructions in $G$ if and only if $Y$ is a minimal hitting set for all \emph{small} obstructions in $G[T_1]$. Notice that $G-T_1$ is a \phcag (we call $T_1$ as an \emph{efficient modulator}, description prescribed in \cref{lem:minimal_hitting}), where for any minimal (or minimum) solution $S$ of size at most $k$, the only purpose of vertices in $S\setminus T_1$ is to hit \emph{large} obstructions. This $T_1$ is the first part of the \emph{nice modulator} $T$ that we want to construct. The other part is $M$, which is a 5-redundant solution (see \cref{def:redundant}) of size $\OO(k^{6})$, which we obtain in polynomial time following the same construction procedure given by \cite{AgrawalM0Z19}. This gives us the additional property that any obstruction of size at least 5 contains at least 5 vertices from $M$ and hence also from $T=T_1\cup M$. We bound the size of such a \emph{nice modulator} $T$ by $\OO(k^{\cyclebound})$.
Next, we analyze the graph $G-T$ and reduce its size by applying various reduction rules. 
\smallskip

For the kernelization algorithm, we look at $G-T$, which is a \phcag and hence  has  a ``\emph{nice clique partition}'' (defined in \cref{sec:clique}). Let $\mathcal{Q}=\{Q_1,Q_2,\ldots\}$ denote such a \emph{nice clique partition} of $G-T$. 
%and $\pi$ be an ordering of $V(H)$.
This partition is similar to the clique partition used by Ke et al.~\cite{KeCOLW18} to design a polynomial kernel for vertex deletion to proper interval graphs.

In the first phase, we bound the size of a clique $Q_i$ for each $Q_i\in\mathcal{Q}$.  Our clique-reduction procedure is based on  ``irrelevant vertex rule'' \cite{Marx10}. 
In particular, we find a vertex that is not necessary for a solution of size at most $k$, and delete it. And after this procedure, we reduce the size of each clique in $G-T$ to $k^{\cO(1)}$.

In the second phase, we bound the size of each connected component in $G-T$. Towards this, we bound the number of cliques in $Q_1,Q_2,\ldots,Q_t$ that contain a neighbor of a vertex in $T$ (say {\em good cliques}). We use \emph{small} obstructions, and in particular, the claw, to bound the number of good cliques by $k^{\cO(1)}$. This automatically divides the clique partition 
into {\em chunks}. A  chunk is  a maximal set of {\em non-good cliques} between a pair of {\em good cliques} where the {\em non-good cliques} along with the {\em good cliques} induce a connected component. We show that  the number of chunks is upper bounded by  $k^{\cO(1)}$. 
Finally, we use a structural analysis to bound the size of each chunk, which includes the design of a reduction rule that computes a minimum cut between the two cliques of a certain distance from the border of the chunk. With this, we bound the number of cliques in each chunk and hence the size of each chunk as well as every connected component by $k^{\cO(1)}$.

In the third and final phase of our kernelization algorithm using the claw obstruction, we bound the number of connected components in $G-T$ by $k^{\cO(1)}$. Using this bound, together with the facts that $|T|\leq k^{\cO(1)}$, and that each connected component is of size $k^{\cO(1)}$, we are able to design a polynomial kernel for \phcad.  We conclude this section by summarizing all the steps in our kernelization algorithm (see \cref{fig:flowchart}).

\begin{figure}[ht!]
	\begin{center}
		\includegraphics[scale=0.6]{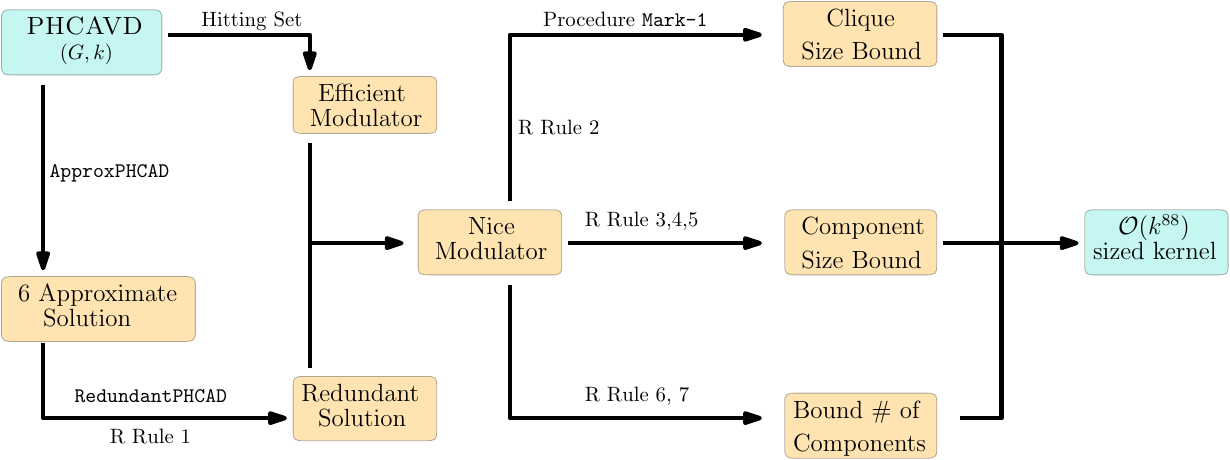}
	\end{center}
	\caption{Flowchart of the Kernelization algorithm for \phcad.}
	\label{fig:flowchart}
\end{figure}

\section{Preliminaries}

\paragraph{Sets and Graph Notations.} We denote the set of natural numbers by $ \mathbb{N} $. For $n \in \mathbb{N}$, by $[n]$ and $[n]_0$, we denote the sets $\{1,2,\cdots, n\}$ and $\{0,1,2,\cdots, n\}$, respectively.  For a graph $G$, $V(G)$ and $E(G)$ denote the set of vertices and edges, respectively.  The neighborhood of a vertex $v$, denoted by $N_G(v)$, is the set of vertices adjacent to $v$. For $A, B \subseteq V (G)$ with $A \cap B = \emptyset$, $ E(A, B)$ denotes the set of edges with one endpoint in $A$ and the other in $B$.  For a set $S \subseteq V(G)$, $G - S$ is the graph obtained by removing $S$ from $G$ and $G[S]$ denotes the subgraph of $G$ induced on $S$. 
A {\em path} $P=v_1,\ldots,v_\ell$ is a sequence of distinct vertices where every consecutive pair of vertices is adjacent. We say that $P$ {\em starts} at $v_1$ and {\em ends} at $v_\ell$. The  vertex set of $P$, denoted by $V(P)$, is the set $\{v_1,\ldots,v_\ell\}$. The {\em internal vertices} of $P$ is the set $V(P) \setminus \{v_1,v_\ell\}$. The {\em length} of $P$ is defined as $|V(P)|-1$. A {\em cycle} is a sequence $v_1,\ldots,v_\ell$ of vertices such that $v_1,\ldots,v_\ell$ is a path and $v_\ell v_1$ is an edge. A cycle (or path) $v_1,\ldots,v_\ell$ is also represented as the ordered set $\{v_1,\ldots,v_\ell\}$. A set $Q \subseteq V(G)$ of pairwise adjacent vertices is called a {\em clique}. A {\em hole} is an induced cycle of length at least four. A vertex is  {\em isolated}  if it has degree zero. For a pair of sets $A, B \subseteq V(G)$, we say $S$ is an $A$-$ B $ cut in $ G $ if there is no edge $ (u,v) $ where $ u \in A\setminus S,~v \in B \setminus S $. Such a $ S $ with minimum cardinality is called as \emph{minimum $ A$-$ B $ cut}. The \emph{distance} between two vertices $u$ and $v$ denoted by $d_G(u,v)$ is the length of a shortest $uv$ path in the graph $G$. The complement graph $\overline{G}$ of a graph $G$ is defined in the same set of vertex $V(G)$ such that  $(u,v) \in E(\overline{G})$ if  and only if $ (u,v) \notin E(G)$. For $ \ell \geq 3 $, we use $ C_{\ell} $ to denote an induced cycle on $ \ell $ vertices; if we add a new vertex to a $ C_{\ell} $ and make it adjacent to none or each vertex in $C_{\ell}$ we end with $ C^*_{\ell} $ or $ W_\ell $, respectively. A \monad is a $ C^*_{\ell} $ with $\ell \geq 4$. We call the $C_{\ell}$ as \mhole and the corresponding isolated vertex as \centre of the \monad. For graph-theoretic terms and definitions not stated explicitly here, we refer to \cite{diestel-book}.

\paragraph{Parameterized problems and kernelization.} A parameterized problem $ \Pi $ is a subset of $ \Gamma^* \times \mathbb{N}$ for some finite alphabet $ \Gamma $. An instance of a parameterized problem consists of $(X, k)$, where $k$ is called the parameter. The notion of kernelization is formally defined as follows. A kernelization algorithm, or in short, a kernelization, for a parameterized problem $ \Pi \subseteq \Gamma^* \times \mathbb{N}$ is an algorithm that, given $(X, k) \in  \Gamma^* \times \mathbb{N}$, outputs in time polynomial in $|X| + k$ a pair $ (X', k')\in  \Gamma^* \times \mathbb{N} $ such that (a) $ (X, k) \in \Pi$ if and only if $(X', k') \in \Pi$ and  (b)  $|x'|,|k| \leq g(k)$, where $ g $ is some computable function depending only on $ k $. The output of kernelization $  (X', k') $ is referred to as the kernel and the function $ g $ is referred to as the size of the kernel. If $ g(k) \in k^{\OO(1)}  $ , then we say that $ \Pi $ admits a polynomial kernel.  We refer to the monographs \cite{DBLP:series/mcs/DowneyF99,DBLP:series/txtcs/FlumG06,DBLP:books/ox/Niedermeier06} for a detailed study  of the area of kernelization.

\paragraph{Proper Helly  Circular-arc Graphs.} A {\em proper Helly circular-arc graph} is an intersection graph of a set of arcs on a circle such that none of the arcs properly contains another (proper) and every set of pairwise intersecting arcs has a common intersection (Helly).  The following is a characterization of proper Helly circular arc graphs.

\begin{figure}[t!]
	\begin{center}
		\includegraphics[scale=0.6]{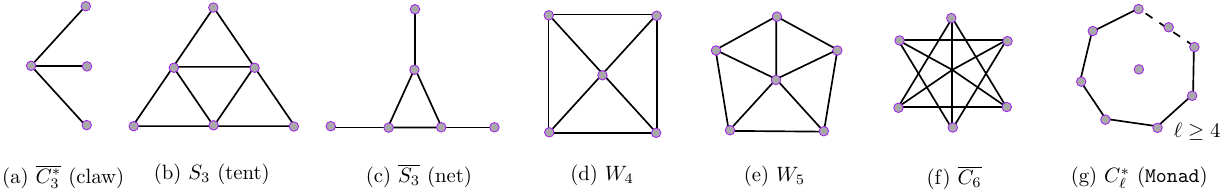}
	\end{center}
	\caption{Forbidden induced subgraphs of \phcags.}
	\label{fig:obs}
\end{figure}

% Any proper Helly circular-arc graph does not contain any of the following graphs, called {\em obstructions}, as an induced subgraphs \cite{DBLP:journals/dam/LinSS13}.

% \begin{itemize}
%     \item[$\bullet$] \textbf{Claw.} A complete bipartite graph $K_{1,3}$.
    
%     \item[$\bullet$] \textbf{Tent.}
    
%     \item[$\bullet$] \textbf{Net.}
    
%     \item[$\bullet$] \textbf{4-Wheel.} A complete bipartite graph $K_{1,3}$.
    
%     \item[$\bullet$] \textbf{5-Wheel.} A complete bipartite graph $K_{1,3}$.
    
%     \item[$\bullet$] \textbf{6-Cycle Complement.}
    
%     \item[$\bullet$] \textbf{Global Hole.}

% \end{itemize}

% \begin{definition}[Proper Helly  Circular-arc Graphs]\label{defn:circular-arc} \todo{cite here}
% \emph{	A {\em proper Helly circular-arc graph} is an intersection graph of a set of arcs on a circle such that none of the arcs	properly contains another (proper) and every	set of pairwise intersecting arcs has a common intersection (Helly).}
	
% 	\end{definition}

%\begin{theorem}[\cite{DBLP:journals/dam/LinSS13}] 
%	A graph is a proper Helly circular-arc graph if and only if 
%	it contains no $ \overline{C_3^*},~
%	S_3,~ \overline{S_3},~W_4,~W_5,~\overline{C_6}$ or $C^*_{\ell} $  for $ \ell \geq 4 $.
%	\end{theorem}

\begin{proposition}[\cite{DBLP:journals/dam/LinSS13}] 
	A graph is a proper Helly circular-arc graph if and only if 
		it contains neither claw, net, tent, wheel of size 4, wheel of size 5, complement of cycle of length 6, nor \monad of length at least 4 as	induced subgraphs.
	\end{proposition}

 \begin{proposition}[Theorem 1.3~\cite{DBLP:journals/corr/abs-2202-00854}]\label{prop:approxAlg}
	\phcad admits a polynomial-time 6-approximation algorithm, called \appIVD.
\end{proposition}
\begin{sloppypar}{\noindent \textit{Nice Clique Partition.}} \label{sec:clique}
	For a connected graph $ G $, a clique partition $\cd= ( Q_1, Q_2, \ldots, Q_{|\cd|}(=Q_0)) $ is called  a  \emph{nice clique partition} of  $G$ if (i) $\bigcup_{i} V(Q_i)= V(G) $, 
	%(ii) $\bigcup_{i} E(Q_i)= E(G)$, 
	(ii) $V(Q_i) \cap V(Q_j) = \emptyset$ if $i \neq j$, and (iii) $E(Q_i, Q_j) = \emptyset$ if $|i-j| > 1$   holds. In such a \emph{nice clique partition} every edge of $G$ is either inside a clique in $\mathbb{Q}$ or present between vertices from adjacent cliques. For a proper circular-arc graph such a partition always exists and can be obtained in $n^{\OO(1)} $ time using a procedure similar to that for a proper interval graph \cite{DBLP:conf/latin/0001S0Z18}.
	 \end{sloppypar}

\section{Constructing an Efficient Modulator}

 We classify the set of obstructions for \phcags as follows. Any obstruction of size less than $\cyclebound$  is known as a {\em small} obstruction, while other obstructions are said to be \emph{large}. In this section we construct an \emph{efficient modulator} $T_1$, of size $\OO(k^{\cyclebound})$ such that $G-T_1$ is a \phcag  with some additional properties that are mentioned in later part.
 
\begin{proposition}[Lemma 3.2~\cite{DBLP:journals/siamdm/FominSV13}] \label{lem:hitting_set}
Let $\mathcal{F}$ be a family of sets of cardinality at most $d$ over a universe $U$ and let $k$ be a
positive integer. Then there is an $\mathcal{O}(|\mathcal{F}|(k + |\mathcal{F}|))$ time algorithm that finds a non-empty family of sets
$\mathcal{F}' \subseteq \mathcal{F}$ such that
\begin{enumerate}
    \item  For every $Z \subseteq U$ of size at most $k$, $Z$ is a minimal hitting set of $\mathcal{F}$ if and only if $Z$ is a
minimal hitting set of $\mathcal{F}'$; and
\item $|\mathcal{F}'| \leq d!(k + 1)^d$.
\end{enumerate}
\end{proposition}

Using \cref{lem:hitting_set} we identify a vertex subset of $V (G)$, which allows us to forget about \emph{small} obstructions in $G$ and  concentrate on \emph{large} obstructions for the kernelization algorithm for \phcad. 

\begin{lemma} \label{lem:minimal_hitting}
Let $(G, k)$ be an instance of \phcad. In polynomial-time,  either we conclude that  $(G, k)$ is a \no-instance, or we can construct a vertex subset $T_1$ such that
\begin{enumerate}
    \item  Every set $Y \subseteq V (G)$ of size at most $k$ is a minimal hitting set for all \emph{small} obstructions in $G$ if and only if it is a minimal hitting set for all \emph{small} obstructions in $G[T_1]$; and
    \item $ |T_1| \leq  12!{(k + 1)}^{12}+6k $.
\end{enumerate}
\end{lemma}

\begin{proof}
	Using \cref{prop:approxAlg}, in polynomial-time we construct a $6$-approximate solution $T'$.
	We also construct $\mathcal{F}_G,U_G$ where $U_G$ consists of all the vertices in $G$ while
	$\mathcal{F}_G$ contains every minimal set of vertices in $G$ that induces a {\em small} obstruction. Applying \cref{lem:hitting_set} on $\mathcal{F}_G,U_G$, in polynomial-time we construct a vertex set $T''$. If $|T'|>6k$ or 
	$|T''|>(\cyclebound+1)!(k + 1)^{\cyclebound}$, we conclude that $(G,k)$ is a \no-instance.
	Otherwise, we have a modulator $T_1=T'\cup T''$  of size $\OO(k^{\cyclebound})$, such that $G-T_1$ is a \phcag, and every set $Y \subseteq V (G)$ of size at most $k$ is a minimal hitting set of all \emph{small} obstructions in $G$ if and only if it is a minimal hitting set for all \emph{small} obstructions in $G[T_1]$. \qed
\end{proof}

Let $S$ be a minimal (or minimum) solution of size at most $k$. Then, the only purpose of the vertices in $S\cap (V(G)\setminus T_1)$ is to hit \emph{large} obstructions. We call the modulator constructed above an {\em efficient modulator}.  We summarize these discussions in the next lemma.

\begin{lemma}\label{lem-efficient}
Let $(G, k)$ be an instance of \phcad.  In polynomial time, we can either construct an \emph{efficient} modulator $T_1\subseteq V(G)$ of size $\OO(k^{\cyclebound})$, or conclude that $(G, k)$ is a \no-instance.
\end{lemma}

\section{Computing a Redundant Solution}\label{sec:redundant}

In this section, our main purpose is to prove \cref{lem:redundant}. Intuitively,  this lemma asserts that in $n^{\OO(1)}$ time we can compute an $r$-redundant solution $M$ whose size is polynomial in $k$ (for a fixed constant $r$). Such a set $M$ plays a crucial role in many of the reduction rules that follow this section while designing our kernelization algorithm. We remark that in this section we use the letter $\ell$ rather than $k$ to avoid confusion, as we will use this result with $\ell=k+2$. Towards the definition of redundancy, we require the following notions and definitions.

\begin{definition}[$t$-solution]
	Let $(G,k)$ be an instance of \phcad. A subset $S\subseteq V(G)$  of size at most $t$ such that $G- S$ is a \phcag  is called a {\em $t$-solution}.
\end{definition}
\begin{definition}[$t$-necessary]
	A family ${\cal W}\subseteq 2^{V(G)}$ is called  {\em $t$-necessary} if and only if every {\em $t$-solution} is a hitting set for ${\cal W}$.
\end{definition}

Given a family ${\cal W}\subseteq 2^{V(G)}$, we say that an obstruction $\mathbb{O}$ is {\em covered by $\cal W$} if there exists $W\in{\cal W}$, such that $W\subseteq V(\mathbb{O})$. 

\begin{definition}[$t$-redundant]\label{def:redundant}{\rm
		Given a family ${\cal W}\subseteq 2^{V(G)}$ and $t\in\mathbb{N}$, a subset $M\subseteq V(G)$ is {\em $t$-redundant with respect to $\cal W$} if for every obstruction $\mathbb{O}$ that is not covered by $\cal W$, it holds that $|M\cap V(\mathbb{O})|>t$.}
\end{definition}

\begin{definition}\label{def:copy}{\rm
		Let $G$ be a graph, $U\subseteq V(G)$, and $t\in\mathbb{N}$. Then, $\copyg(G,U,t)$ is defined as the graph $G'$ in the vertex set $V(G)\cup\{v^i \mid v\in U, i\in[t]\}$ and the edge set $E(G)\cup\{(u^i,v) \mid (u,v)\in E(G), u \in U,i \in[t]\}\cup \{(u^i,v^j) \mid (u,v)\in E(G), u,v\in U, i,j\in[t]\}\cup\{(v,v^i) \mid v\in U, i\in[t]\}\cup\{(v^i,v^j) \mid v\in U, i,j\in[t], i\neq j\}$.}
\end{definition}

Informally, $\copyg(G,U,t)$ is simply the graph $G$ where for every vertex $u\in U$, we add $t$ twins that (together with $u$) form a clique. Intuitively, this operation allows us to make a vertex set ``undeletable''; in particular, this enables us to test later whether a vertex set is ``redundant'' and therefore we can grow the redundancy of our solution, or whether it is ``necessary'' and hence we should update $\cal W$ accordingly. Before we turn to discuss computational issues, let us first assert that the operation in \cref{def:copy} does not change the class of the graph, which means it remains a \phcag. We verify this in the following lemma.

\begin{lemma}\label{lem:addCliquTwins}
	Let $G$ be a graph, $U\subseteq V(G)$, and $t\in\mathbb{N}$. If $G$ is a \phcag, then $G'=\copyg(G,U,t)$ is also a \phcag.
\end{lemma}

\begin{proof}
	Suppose that $G$ is a \phcag. Then $G$ admits a proper circular-arc representation and has Helly property (no three arcs cover the circle \cite{DBLP:journals/corr/abs-2202-00854}) i.e. all its vertices can be presented as arcs on a circle $C$. Notice that the newly introduced vertices in $G'$ are twin (copy) vertices of $G$. These twin vertices are given the same arc representations on C as the original vertices in $G$. It is easy to see that this indeed is a proper circular-arc representation of $G'$ with Helly property and hence $G'$ is also a \phcag.\qed
\end{proof}

Now, we present two simple claims that exhibit relations between the algorithm \appIVD\ and  \cref{def:copy}. After presenting these two claims, we will be ready to give our algorithm for computing a redundant solution. Generally speaking, the first claim shows the meaning of a situation where \appIVD returns a ``large'' solution; intuitively, for the purpose of the design of our algorithm, we interpret this meaning as an indicator to extend ${\cal W}$. 

\begin{lemma}\label{lem:approxFails}
	Let $G$ be a graph, $U\subseteq V(G)$, and $\ell\in\mathbb{N}$. If the algorithm \appIVD\ returns a set $A$ of size larger than $6\ell$ when called with  $G'=\copyg(G,U,6\ell)$ as input, then $\{U\}$ is $\ell$-necessary.
\end{lemma}

\begin{proof}
	Suppose that \appIVD\ returns a set $A$ of size larger than $6\ell$ when called with $G'$ as input. Then, $(G',\ell)$ is a \no-instance. If $(G,\ell)$ is a \no-instance, then trivially, we can say that $\{U\}$ is $\ell$-necessary (as there is no solution of size at most $\ell$, so the statement is vacuously true). Now consider the case when $G$ has an $\ell$-solution $S$ such that $S\cap U=\emptyset$. In particular, $\widehat{G}=G-S$ is a \phcag such that $U\subseteq V(\widehat{G})$. However, this means that $\copyg(\widehat{G},U,6\ell)=G'-S$, which by  \cref{lem:addCliquTwins} implies that $G'-S$ is a \phcag. Thus, $S$ is an $\ell$-solution for $G'$, which is a contradiction (as $(G',\ell)$ is a \no-instance).\qed
\end{proof}

Complementing our first claim, the second claim exhibits the meaning of a situation where \appIVD returns a ``small'' solution $A$; we interpret this meaning as an indicator of growing the redundancy of our current solution $M$ by adding $A$ ---- indeed, this lemma implies that every obstruction is hit one more time by adding $A$ to a subset $U\subseteq M$ (to grow the redundancy of $M$, every subset $U\subseteq M$ will have to be considered). 

\begin{tcolorbox}
	\textbf{\underline{Algorithm 2: \redIVD$(G, \ell, r)$}}
	
	\begin{enumerate}
		\item Initialization: \\$M_0:=  \appIVD(G)$, \\ ${\cal W}_0:=\emptyset$, \\${\cal T}_0:=\{(v) \mid v\in M_0\}$.
		
		\item 
		If $|M_0|>6\ell$, return ``$(G,\ell)$ is a \no-instance''.
		
		Otherwise, $i=1$ and go to Step 3.
		
		\item While  $i\leq r$, for every tuple $(v_0,v_1,\ldots,v_{i-1})\in {\cal T}_{i-1}$:
		\begin{enumerate}
			\item  $A:= \appIVD (\copyg(G,\{v_0,v_1,\ldots,v_{i-1}\},6\ell))$.
			\item If $|A|>6\ell$, $~{\cal W}_i:= {\cal W}_{i-1} \cup \{\{v_0,v_1,\ldots,v_{i-1}\}\}$.
			\item Otherwise, \\ ${ M}_i:= { M}_{i-1} \cup \{u \mid u\in (A\cap V(G))\setminus \{v_0,v_1,\ldots,v_{i-1}\}\}$, \\
			${\cal T}_i~~:= {\cal T}_{i-1} \cup 
			\{(v_0,v_1,\ldots,v_{i-1},u) \mid u\in (A\cap V(G))\setminus \{v_0,v_1,\ldots,v_{i-1}\}\}$.
			\item $i=i+1;$
		\end{enumerate}
		
		\item Return $(M_r,{\cal W}_r)$.
	\end{enumerate}
\end{tcolorbox}

\begin{lemma}\label{lem:approxSucceeds}
	Let $G$ be a graph, $U\subseteq V(G)$, and $\ell\in\mathbb{N}$.
	If the algorithm \appIVD\ returns a set $A$ of size at most $6\ell$ when called with $G'=\copyg(G,U,6\ell)$ as input, then for every obstruction $\mathbb{O}$ of $G$, $|V(\mathbb{O})\cap U|+1\leq |V(\mathbb{O})\cap (U\cup (A\cap V(G)))|$.
\end{lemma}

Now, we describe our algorithm, \redIVD, which computes a redundant solution. First, \redIVD\ initializes $M_0$ to be the 6-approximate solution to \phcad with $(G,\ell)$ as input, ${\cal W}_0:=\emptyset$ and ${\cal T}_0:=\{(v) \mid v\in M_0\}$. If $|M_0|>6\ell$, then \redIVD\ concludes that $(G,\ell)$ is a \no-instance. Otherwise, for $i=1,2,\ldots,r$ (in this order), the algorithm executes the following steps (Step 3 in the figure below) and eventually, it outputs the pair $(M_r,{\cal W}_r)$. In the \redIVD algorithm, by \appIVD($H$) we mean the $6$-approximate solution returned by the approximation algorithm to the input graph $H$.

Let us comment that in this algorithm we make use of the sets ${\cal T}_{i-1}$ rather than going over all subsets of size $i$ of $M_{i-1}$ in order to obtain a substantially better algorithm in terms of the size of the redundant solution produced.

The properties of the algorithm \redIVD\ that are relevant to us are summarized in the following lemma and observation, which are proved by induction and by making use of Lemmata \cref{lem:addCliquTwins}, \cref{lem:approxFails} and \cref{lem:approxSucceeds}. Roughly speaking, we first assert that, unless $(G,\ell)$ is concluded to be a \no-instance, we compute sets ${\cal W}_i$ that are $\ell$-necessary as well as that the tuples in ${\cal T}_i$ ``hit more vertices'' of the obstructions in the input as $i$ grows larger.

\begin{lemma}\label{lem:redAlgCor}
	Consider a call to \redIVD\ with $(G,\ell,r)$ as input that did not conclude that $(G,\ell)$ is a \no-instance. For all $i\in[r]_0$, the following conditions hold:
	\begin{enumerate}
		\item\label{cond1} For any set $W\in{\cal W}_i$, every solution $S$ of size at most $\ell$ satisfies $W\cap S\neq\emptyset$.
		\item\label{cond2} For any obstruction $\mathbb{O}$ of $G$ that is not covered by ${\cal W}_i$, there exists $(v_0,v_1,\ldots,v_i)\in{\cal T}_i$ such that $\{v_0,v_1,\ldots,v_i\}\subseteq V(\mathbb{O})$.
	\end{enumerate}
\end{lemma}

\begin{sloppypar}
	\begin{proof}
		The proof is by induction on $i$. In the base case, where $i=0$,  Condition \cref{cond1} trivially holds as ${\cal W}_0=\emptyset$, and thus there are no sets in ${\cal W}_0$.  Condition \cref{cond2} holds as $M_0$ is a solution (so each obstruction must contain at least one vertex from $M_0$) and ${\cal T}_0$ simply contains a 1-vertex tuple for every vertex in $M_0$. Now, suppose that the claim is true for $i-1\geq 0$, and let us prove it for $i$.
		
		To prove  Condition \cref{cond1}, consider some set $W\in{\cal W}_i$. If $W\in{\cal W}_{i-1}$, then by the inductive hypothesis, every solution of size at most $\ell$ satisfies $W\cap S\neq\emptyset$. Thus, we next suppose that $W\in{\cal W}_i\setminus{\cal W}_{i-1}$. Then, there exists a tuple $(v_0,v_1,\ldots,v_{i-1})\in {\cal T}_{i-1}$ in whose iteration \redIVD\ inserted $W=\{v_0,v_1,\ldots,v_{i-1}\}$ into ${\cal W}_i$. In that iteration, \appIVD\ was called with ~$\copyg(G,W,6\ell)$ as input, and returned a set $A$ of size larger than $6\ell$. Thus, by  \cref{lem:approxFails}, every solution $S$ of size at most $\ell$ satisfies $W\cap S\neq\emptyset$.
		
		To prove  Condition \cref{cond2}, consider some obstruction $\mathbb{O}$ of $G$ that is not covered by ${\cal W}_i$. By the inductive hypothesis and since ${\cal W}_{i-1}\subseteq{\cal W}_i$, there exists a tuple $(v_0,v_1,\ldots,v_{i-1})\in{\cal T}_{i-1}$ such that $\{v_0,v_1,\ldots,v_{i-1}\}\subseteq V(\mathbb{O})$. Consider the iteration of \redIVD\ corresponding to this tuple, and denote $U=\{v_0,v_1,\ldots,v_{i-1}\}$. In that iteration, \appIVD\ was called with $\copyg(G,U,6\ell)$ as input, and returned a set $A$ of size at most $6\ell$. By  \cref{lem:approxSucceeds}, $|V(\mathbb{O})\cap U|+1\leq |V(\mathbb{O})\cap (U\cup (A\cap V(G)))|$. Thus, there exists $v_i\in (A\cap V(G))\setminus U$ such that $U\cup\{v_i\}\subseteq V(\mathbb{O})$. However, by the specification of \appIVD, this means that there exists $(v_0,v_1,\ldots,v_i)\in{\cal T}_i$ such that $\{v_0,v_1,\ldots,v_i\}\subseteq V(\mathbb{O})$.\qed
	\end{proof}
\end{sloppypar}

Towards showing that the output set $M_r$ is ``small'', let us upper bound the sizes of the sets $M_i$ and ${\cal T}_i$.

\begin{observation}\label{obs:redAlgTime}
	Consider a call to \redIVD\ with $(G,\ell,r)$ as input that did not conclude that $(G,\ell)$ is a \no-instance. For all $i\in[r]_0$, $|M_i|\leq\sum_{j=0}^i(6\ell)^{j+1}$, $|{\cal T}_i|\leq (6\ell)^{i+1}$ and every tuple in ${\cal T}_i$ consists of distinct vertices.
\end{observation}

\begin{proof}
	The proof is by induction on $i$.  In the base case, where $i=0$, the correctness follows as \appIVD\ returned a set of size at most $6\ell$. Now, suppose that the claim is true for $i-1\geq 0$, and let us prove it for $i$. By the specification of the algorithm and inductive hypothesis, we have that $|M_i|\leq |M_{i-1}| + 6\ell|{\cal T}_{i-1}|\leq \sum_{j=1}^{i+1}(6\ell)^j$ and $|{\cal T}_i|\leq 6\ell|{\cal T}_{i-1}|\leq (6\ell)^{i+1}$. Moreover, by the inductive hypothesis, for every tuple in ${\cal T}_i$, the first $i$ vertices are distinct, and by the specification of \appIVD, the last vertex is not equal to any of them.\qed
\end{proof}

By the specification of \redIVD, as a corollary to  \cref{lem:redAlgCor} and  \cref{obs:redAlgTime}, we directly obtain the following result.

\begin{corollary}\label{cor:redAlg}
	Consider a call to \redIVD\ with $(G,\ell,r)$ as input that did not conclude that $(G,\ell)$ is a \no-instance. For all $i\in[r]_0$, ${\cal W}_i$ is an $\ell$-necessary family and $M_i$ is a $\sum_{j=0}^{i}(6\ell)^{j+1}$-solution that is $i$-redundant with respect to ${\cal W}_i$.
\end{corollary}

\begin{lemma}\label{lem:redundant}
	Let $r\in\mathbb{N}$ be a fixed constant, and $(G,\ell)$ be an instance of \phcad. In polynomial-time, it is possible to either conclude that $(G,\ell)$ is a \no-instance, or compute an $\ell$-necessary family ${\cal W}\subseteq 2^{V(G)}$ and a set $M\subseteq V(G)$, such that ${\cal W} \subseteq 2^M$ and $M$ is a $(r+1)(6\ell)^{r+1}$-solution that is $r$-redundant with respect to~$\cal W$.
\end{lemma}

\begin{proof}
	Clearly, \redIVD\ runs in polynomial-time (as $r$ is a fixed constant), and by the correctness of \appIVD, if it concludes that $(G,\ell)$ is a \no-instance, then this decision is correct. Thus, since $\sum_{i=0}^r(6\ell)^{i+1}\leq (r+1)(6\ell)^{r+1}$, the correctness of  \cref{lem:redundant} now directly follows as a special case of  \cref{cor:redAlg}. Thus, our proof of \cref{lem:redundant} is complete.\qed
\end{proof}

In light of  \cref{lem:redundant}, from now on, we suppose that we have an $\ell$-necessary family ${\cal W}\subseteq 2^{V(G)}$ along with a $(r+1)(6\ell)^{r+1}$-solution $M$ that is $r$-redundant with respect to $\cal W$ for $r=\redundC$.
Let us note that, any obstruction in $G$ that is not covered by $\calW$ intersects $M$ in at least six vertices.
We have the following reduction rule that follows immediately from \cref{lem:redAlgCor}.

\begin{reduction rule}\label{rr:w-size-not-one} Let $v$ be a vertex such that $\{v \} \in \calW$. Then, output the instance $(G-\{v\},k-1)$.
\end{reduction rule}

From here onwards we assume that each set in $\cal W$ has a size at least $2$. 

\smallskip

\noindent
\textbf{Nice Modulator.} Once we construct both the \emph{efficient modulator} $T_1$ and \emph{redundant solution} $M$, we take their union and consider that set of vertices as a modulator, we called it as \emph{nice modulator}. 

From here onwards, for the remaining sections,  we assume that
\medskip

\noindent\fbox{
	\begin{minipage}{0.96\textwidth}\label{nicemod}
		We have a \emph{nice modulator}  $T \subseteq V(G)$ along with $(k+2)$-necessary family $\calW \subseteq 2^T$  satisfying the following:
		\vspace{-2mm}
		\begin{itemize}
			\item[$\bullet$] $G -T$ is a \phcag.
			\item[$\bullet$]$|T| \leq \OO(k^{12})$.
			\item[$\bullet$] For any \emph{large} obstruction $\obs$  containing no  $W\in \calW$, we have $|V(\obs) \cap T| \geq 6$.
		\end{itemize} 
	\end{minipage}
}

\section{Bounding the Size of  each Clique}

In this section, we consider a \emph{nice modulator} $T$ of $G$ obtained in the previous section and we bound the size of each clique in a \emph{nice clique partition} $\mathcal{Q}$ of $G-T$  in polynomial time.  If there is a \emph{large} clique in $\mathcal{Q}$  of size more than $\mathcal{O}(k^{12})  $, we can safely find and remove an \emph{ irrelevant vertex} from the clique, thus reducing its size. Next, we prove a simple result that will later be used to bound the size of each clique in $ G -T $.

\begin{lemma}\label{lem:small_obs}
	Let $H$ be an induced path in $G$. Consider a vertex $v \in V(G) \setminus V(H)$. If $ v $ has more than four neighbors in $ V(H) $ then $ G[V(H) \cup \{v\}] $  contains a \emph{small} obstruction (claw).	
\end{lemma}

\begin{proof}
	Assume that $\card{N(v) \cap V(H)} \geq 5$. Let $H$ be an induced path from $x$ to $y$ for some $x, y \in V(G)$. Let $v_1, v_2, v_3, v_4, v_5 \in V(H)$ be the first $5$ neighbours of $v$ that appear as we traverse $H$ from $x$ to $y$. Since the path is induced so $ (v_1, v_3), (v_3, v_5) \notin E(G) $. So $\set{v, v_1, v_3, v_5}$ induces a $ \overline{C_3^*} $ (claw), which is a \emph{small} obstruction.\qed
\end{proof}

\paragraph{Marking Scheme.} We start with the following marking procedure, which marks $k^{\OO(1)}$ vertices in each clique  $\cl \in \mathcal{Q}$.

We will now bound the size of the set $ T(\cl)$.

\begin{remark}\label{rem:Mark1}
	Observe that the procedure \markone can be executed in polynomial time. Also, note that $\card{T(\cl)} \leq 2(k+1) \card{T}^4$. 
\end{remark}

\begin{reduction rule}\label{rule:bound_clique}
	If there exists a vertex $v \in \cl \setminus  T(\cl)$ for some clique  $\cl  \in \mathcal{Q} \subseteq V(G) \setminus T$, then delete $v$. 
\end{reduction rule}

\begin{lemma}\label{lem:cliqueRed}
	Reuction Rule \ref{rule:bound_clique} is safe. 
\end{lemma}

\begin{proof}
	Consider an application of reduction rule \ref{rule:bound_clique} in which a vertex, say $v \in \cl \setminus  T(\cl)$ was deleted from some clique  $\cl  \in \mathcal{Q}$. we claim the following.
	\begin{claim}
		$(G,k)$ is a \yes-instance of \phcad if and only if $(G-v, k)$ is a \yes-instance of \phcad.
	\end{claim}
	
	\noindent
	$ (\Rightarrow) $ If $(G, k)$ is a \yes-instance, then so is $(G-v, k)$, since $G-v$ is an induced subgraph of $G$. 
	
	\noindent $ (\Leftarrow) $ To prove the other direction we use contradiction. 
	Suppose $(G-v,k)$ is a \yes-instance but $ (G,k) $ is not. And let $X \subseteq V(G-v)$ be a solution of size at most $k$. That is $(G - v) - X$ is a \phcag. Since $(G, k)$ is a \no-instance, $G-X$ can not be a \phcag. Hence $G - X$ must contain an obstruction, say, $ \obs $. 	Clearly, $v$ must be a vertex in  $ V(\obs)$, otherwise, $\obs$ would also be an obstruction in $(G-v) - X$, which contradicts the fact that $(G - v) - X$  is a \phcag. 
	
	We first claim that $ \obs $ is a \emph{large} obstruction.  Suppose it is not, i.e. $ \obs $ is a \emph{small} obstruction. Note that $X$ hits all obstructions in $G - v$, and  $G[T]$ is a subgraph of $G - v$ as $v \notin T$. So $X$ also hits all obstructions in $G[T]$, in particular, also all \emph{small} obstructions in $G[T]$. Let $Y \subseteq X$ be a minimal hitting set for all \emph{small} obstructions in $G[T]$. Then, by the definitions of $T$ and $Y$, we can conclude that $Y$ hits all \emph{small} obstructions in $G$ as well (using \cref{lem:minimal_hitting}). But since $ \obs $ is an obstruction contained in $G - X$ and $Y \subseteq X$, $\obs$ has to be a  \emph{large} obstruction in $G - Y$, a contradiction. 	Thus, $ \obs$ being a \emph{large} obstruction in $G - X$ must be ${C_{\ell}}^*$ (\monad) where $\ell>12$. Also $ v\in V(\obs) $. So there is no \emph{small} obstruction containing $v$ in $G-X$.
	
	Next, we claim that such an obstruction
	$ \obs $ can not contain any  $W\in\calW$. As $X\cup \{v\}$ is a $(k+1)$-solution for $G$, $X\cup \{v\}$ is a hitting set for the $(k+2)$-necessary family $\calW$. But $\obs\cap(X\cup \{v\}) \neq \emptyset$. This implies that $\obs$ does not contain any $W\in\calW$. So, $M$ and hence $T$ contains at least five vertices from $\obs$ i.e. $|V(\obs)\cap T|\geq 5$.

	To show equivalence between the instances $(G,k)$ and $(G-v,k)$, we either find an obstruction $\obs'$ in $(G-v) - X$ or we show that $v$ is a part of \emph{small} obstruction in $G-X$.  We argue for  the following two cases  depending on the nature of $ v $ in $ V(\obs) $. We use $ \prv(v) $ and $\nxt(v)$ to denote the adjacent vertices of $ v $ in  $ V(\obs) $ (here the selection is arbitrary). Let $ \glb $ denote the \centre  of this \monad $ \obs $.

	\begin{description}
		\item[Case A]  Here we consider the case when vertex $ v $ is not the \centre of the \monad $ \obs:= C^*_{\ell}$, i.e., $ v \neq \glb $.	We argue for all the following eight cases  depending on whether the vertices $ \prv(v), \nxt(v)$, and  $ \glb $ of $ \obs $ belong to $T $ or not. 	Notice that $ v $ was deleted because it was an irrelevant (unmarked) vertex. From the redundant solution property (\cref{lem:redundant}), we know that $ \obs $ has at least \textbf{five}  vertices from $ T $ and $ v $ is adjacent to at most two of them while non-adjacent to the rest.
	\end{description}

	\begin{enumerate}
		
		\item $ \boldsymbol{\prv(v) \in T,~\nxt(v) \in T,~\glb \in T} $.  Let $ u $ be a vertex in $ \obs \cap T $ such that $ (u,v) \notin E(G) $. Note that such a vertex $ u $  always exists because of the redundant  solution property. During Procedure \markone, we have added a set $ S: =\markk_i[A,B]$ of at least $2(k+1)$ vertices from $ \cl $ where $ A= \{  \prv(v), \nxt(v)\} $ and $ B= \{ u , \glb\} $. Otherwise, we would have added $v$  to $ T(\cl) $. So each vertex in $ S $ is  non-adjacent  to both $ u $ and $\glb  $ and adjacent to both $ \prv(v) $ and $\nxt(v) $.  Since $ |S| > k $, there  exists a vertex in $ S $ which is not in $ X $. Let $ v' $  be such a vertex (arbitrarily chosen) from $ S \setminus X $. 		
		Assume that $u_1$ and $u_2$ are the two closest vertices of $u$ along the clockwise and anti-clockwise directions, respectively in $V(\obs)$, which are  also adjacent to $v'$. Notice that there is an induced path  $ P $  between $ u_1 $ and $ u_2 $ passing through $ u $ such that    $ V(P) \subseteq V(\obs) $ and $ N(v') \cap V(P -u_1 -u_2) = \emptyset$. Clearly, $ v \notin P $. Let $ C $ be the cycle induced by the vertices $ V(P) \cup \{v'\} $. Since $ v' \notin X $ and $ P \cap X =\emptyset$, $C \cup \{ \glb\}$ forms an obstruction that is contained in $(G-v)-X$.  And this contradicts the fact that $X$ is a solution to $G-v$ for \phcad.
		\smallskip
		
		\begin{sloppypar}
			
			\item $\boldsymbol{ \prv(v) \in T,~\nxt(v) \in T,~\glb \notin T }$. Let $ u $ be a vertex in $ \obs \cap T $ such that $ (u,v) \notin E(G) $. Recall that  during Procedure \markone, we have added a set $ S: =\markk_i[A,B]$ of at least $2(k+1)$ vertices from $ \cl $ where $ A= \{  \prv(v), \nxt(v)\} $ and $ B= \{ u\} $. Otherwise, we would have added $v$  to $ T(\cl)$. But each vertex  in $ S $ is adjacent to both $ \prv(v) $ and $\nxt(v) $.   	
			%	In $ S $, we must add left most $ (k+1) $ vertices, say $ S_1 $ and right most $ (k+1) $ vertices, say $ S_2 $ in $ \cl$. 
			Since $ |S| > k $, there exists a vertex in $ S $ which is not in $ X $. Let $ v' $ be such a vertex (arbitrarily chosen) from $ S \setminus X $. \textbf{If}  $ (v', \glb) \notin E(G) $,  using similar arguments as in case A.1.~we can find an induced subgraph $\obs'$ in $ G $ which  is an obstruction in $(G-v)-X$, a contradiction. \textbf{Else}, $ (v', \glb) \in E(G) $. In this case, we can a find a $ \overline{C_3^*} $, in  $(G-v)-X$ induced by the vertices $\{  v', \prv(v), \nxt(v), \glb \}$, which is an obstruction in $(G-v)-X$, again a contradiction.	
			
		\end{sloppypar}
		
		%		As $ |S_1| > k $ and $ |S_2|>k $ there always exists some vertex in each of $ S_1 $ and $ S_2 $ which are not part of $ X $. Let $ v'_1 $ and $ v'_2 $ be a pair of such vertices  (arbitrarily choosen) of $ S_1 \setminus X $ and $ S_2 \setminus X $, respectively. \textbf{If} either $ (v_1, \glb) \notin E(G) $ or $ (v_2, \glb) \notin E(G) $ using the similar arguments as case A.1.~we can find a induced subgraph $\obs'$ in $ G $ which  is an obstruction in $(G-v)-X$, which contradicts that $X$ is a solution to $G-v$ of \phcad. \textbf{Else}, $ \{v_1, v_2 \} \subseteq N(\glb) $. In this case, we conclude that 
		%		$(G-v)-X$ has an arc representation in which there exists three arcs that cover the cycle, in particular the arcs corresponds to $ v_1 $, $ v_2 $, and $ \glb $ can cover the cycle. Together with \cref{lem:three_arc}, we lead to a contradiction to the fact that $X$ is a solution to $G-v$ of \phcad.
		\smallskip
		
		\item $ \boldsymbol{\prv(v) \in T,~\nxt(v) \notin T,~\glb \in T} $. Let $ u $ be a vertex in $ \obs \cap T $ such that $ (u,v) \notin E(G) $. Note that such a vertex $ u $  always exists because of the redundant solution property. Recall that  during Procedure \markone, we have added a set $ S: =\markk_i[A,B]$ of at least $2(k+1)$ vertices from $ \cl $ where $ A= \{  \prv(v)\} $ and $ B= \{ u, \glb\} $. Otherwise, we would have added $v$  to $ T(\cl)$. 	In $ S $, we must have added the left most $ (k+1) $ vertices, say $ S_1 $ and the right most $ (k+1) $ vertices, say $ S_2 $ in $ \cl$. Since $ |S_1| > k $ and $ |S_2|>k $, there  exists some vertices in each of $ S_1 $ and $ S_2 $ which are not in $ X $. Let $ v'_1 $ and $ v'_2 $ be a pair of such vertices  (arbitrarily chosen) from $ S_1 \setminus X $ and $ S_2 \setminus X $, respectively. 	Since $ v, v'_1, v'_2 \in \cl $ and $ (v,\nxt(v)) \in E(G)$ , either $ v'_1 $ or $ v'_2 $ must be adjacent to $ \nxt(v) $. Without loss of generality, we assume that $ (v'_1,\nxt(v) )  \in E(G)$. Since $ \{\prv(v), \nxt(v) \} \subseteq N(v'_1) $ and  $v'_1 $ is non-adjacent to both $ u $ and $ \glb $, so there must exist  two distinct vertices $ u_1 $ and $ u_2 $ in $ \obs $ such that $ \{u_1,u_2\} \subseteq N(v'_1) $ with an induced path  $ P $  between them passing through $ u $ where   $ V(P) \subseteq V(\obs) $ and $ N(v'_1) \cap V(P -u_1 -u_2) = \emptyset$. Clearly, $ v \notin V(P) $, as $ (v, v'_1) \in E(G)$.  Let $ C $ be the cycle induced by the vertices $ V(P) \cup \{v'_1\}$. Now  $ P \cap X =\emptyset$ and $ v'_1 \notin X $ together imply that $\obs':= C \cup \{ \glb\}$ is an obstruction in $(G-v)-X$, a contradiction.
		
		\smallskip
		
		\item $ \boldsymbol{ \prv(v) \notin T,~\nxt(v) \in T,~\glb \in T }$. Since the selection of $ \prv(v) $ and $\nxt(v)$ is arbitrary, so arguments for this case are similar to that of case A.3.

		\smallskip
		
		\item $ \boldsymbol{ \prv(v) \in T,~\nxt(v) \notin T,~\glb \notin T}$. Let $ u $ be a vertex in $ \obs \cap T $ such that $ (u,v) \notin E(G) $. Note that such a vertex $ u $  always exists because of the redundant solution property. Recall that  during Procedure \markone, we have added a set $ S: =\markk_i[A,B]$ of at least $2(k+1)$ vertices from $ \cl $ where $ A= \{  \prv(v)\} $ and $ B= \{u\} $.  Otherwise, we would have added $v$  to $ T(\cl)$. In $ S $, we must have added the left most $ (k+1) $ vertices, say $ S_1 $ and the right most $ (k+1) $ vertices, say $ S_2 $ in $ \cl$. Since $ |S_1| > k $ and $ |S_2|>k $, there  exist some vertices in each of $ S_1 $ and $ S_2 $ which are not in $ X $. Let $ v'_1 $ and $ v'_2 $ be a pair of such vertices  (arbitrarily chosen) from $ S_1 \setminus X $ and $ S_2 \setminus X $, respectively. As	$ v, v'_1, v'_2 \in \cl $ and $ (v,\nxt(v)) \in E(G)$ , either $ v'_1 $ or $ v'_2 $ must be adjacent to $ \nxt(v) $. Without loss of generality, we assume that $ (v'_1,\nxt(v) )  \in E(G)$. \textbf{If}  $ (v', \glb) \notin E(G) $  using arguments similar to that in case A.3.~, we can find an induced subgraph $\obs'$ in $ G $ which is an obstruction in $(G-v)-X$, a contradiction. \textbf{Else}, $ (v', \glb) \in E(G) $. In this case, we can a find a $ \overline{C_3^*} $, in  $(G-v)-X$ induced by the vertices $\{  v', \prv(v), \nxt(v), \glb \}$, which is an obstruction in $(G-v)-X$, again a contradiction.

		\smallskip
		
		\item $ \boldsymbol{ \prv(v) \notin T,~\nxt(v) \in T,~\glb \notin T }$. Since selection of $ \prv(v) $ and $\nxt(v)$ is arbitrary, so arguments for this case is similar to that of case A.5.
		
		\smallskip
		
		\item $ \boldsymbol{\prv(v) \notin T,~\nxt(v) \notin T,~\glb \in T} $. Let $ u $ be a vertex in $ \obs \cap T $ such that $ (u,v) \notin E(G) $. Note that such a vertex $ u $  always exists because of the redundant solution property. Recall that  during Procedure \markone, we have added a set $ S: =\markk_i[A,B]$ of at least $2(k+1)$ vertices from $ \cl $ where $ A= \emptyset $ and $ B= \{u, \glb\} $. Otherwise, we would have added $v$ as well to $ T(\cl)$. 	\textbf{If} there exists a vertex $ v' \in S \setminus X$ such that $(v',\prv(v)) \notin E(G)$ and $(v',\nxt(v)) \notin E(G)$, then we get an induced subgraph $ \overline{C_3^*} $, in  $G-X$ induced by the vertices $\{  v, v',  \prv(v), \nxt(v)\}$, which is a \emph{small} obstruction in $G-X$. This is  a contradiction to the fact that there is no \emph{small} obstruction containing $ v  $ in $ G -X$. \textbf{If} there exists a vertex $ v' \in S \setminus X$ such that $(v',\prv(v)) \in E(G)$ and $(v',\nxt(v)) \in E(G)$, then using the same procedure as in case A.1.~, we can find an induced subgraph $\obs'$ in $ G $ which  is an obstruction in $(G-v)-X$, a contradiction. \textbf{Else}, each vertex in $ S $ is adjacent to exactly one of  $ \prv(v) $ and  $ \nxt(v) $. During the procedure \texttt{Mark-1}, in $ S $ we must have added the left most $ (k+1) $ vertices, say $ S_1 $ and the right most $ (k+1) $ vertices, say $ S_2 $ in $ \cl$. Since $ |S_1| > k $ and $ |S_2|>k $, there exist some vertices in each of $ S_1 $ and $ S_2 $ which are not in $ X $. Let $ v'_1 $ and $ v'_2 $ be a pair of such vertices  (arbitrarily chosen) from $ S_1 \setminus X $ and $ S_2 \setminus X $, respectively.  Without loss of generality,  we assume that $ (v'_1,\prv(v) )  \in E(G)$ and $ (v'_2,\nxt(v) )  \in E(G)$. Clearly, $ (v'_2,\prv(v) )  \notin E(G)$ and $ (v'_1,\nxt(v) )  \notin E(G)$. Since $ u$ is non-adjacent to both $ v'_1$ and $v'_2$, there exists a pair of vertices $ u_1 $ and $ u_2 $ 	in $ \obs $ such that $ (u_1,v'_1) \in E(G) $, $ (u_2,v'_2) \in E(G) $ with an induced path  $ P \subseteq \obs $ between $ u_1 $ and $ u_2 $ passing through $ u $ where $ N(\{v'_1, v'_2 \}) \cap V(P -u_1 -u_2) = \emptyset$. Clearly, $ v \notin P $, as $ (v, v'_1), (v, v'_2) \in E(G)$.  Let $ C $ be the cycle induced by the vertices $ V(P) \cup \{v'_1, v'_2\} $. Now  $ P \cap X =\emptyset$, $ v'_1, v'_2 \notin X $, $ N(\glb) \cap \{v'_1, v'_2\} = \emptyset $ together imply that $\obs':= C \cup \{ \glb\}$ is an obstruction in $(G-v)-X$, a contradiction.
		
		\smallskip
		
		\item $ \boldsymbol{\prv(v) \notin T,~\nxt(v) \notin T,~\glb \notin T }$.  Let $ u $ be a vertex in $ \obs \cap T $ such that $ (u,v) \notin E(G) $. Note that such a vertex $ u $  always exists because of the redundant solution property. But during Procedure \markone, we have added a set $ S: =\markk_i[A,B]$ of at least $2(k+1)$ vertices from $ \cl $ where $ A= \emptyset $ and $ B= \{u\} $. Otherwise,  $v$ would have been added  to $ T(\cl)$. 
		
		\textbf{If} there exists a vertex $ v' \in S \setminus X$ such that  $ v' $ is non-adjacent to both  $  \prv(v)$ and $ \nxt(v)$, then we can a find a $ \overline{C_3^*} $ in  $G-X$ induced by the vertices $\{ v, v', \prv(v), \nxt(v)\}$, which is a \emph{small} obstruction  containing $v$ in $G-X$, a contradiction.	
		
		\textbf{If} there exists a vertex $ v' \in S \setminus X$ such that $ (v',\glb) \in E(G) $ and $ v' $ is adjacent to both $  \prv(v)$ and $ \nxt(v)$, then we can a find a $ \overline{C_3^*} $ induced by the vertices $\{  v', \prv(v), \nxt(v), \glb \}$, which is an obstruction in $(G-v)-X$, a contradiction.

		\textbf{If} there exists a vertex $ v' \in S \setminus X$ such that $ (v',\glb) \in E(G) $ and $ v' $ is adjacent to exactly one of  $  \prv(v)$ and $ \nxt(v)$, then we argue as follows. Without loss of generality, we assume that $ (v',\prv(v) )  \in E(G)$. Suppose $ v^*$ is the adjacent vertex of $ \prv(v) $ other than $ v $ in the obstruction.  When $ (v', v^*) \notin E(G) $, we  can a find a $ \overline{S_3} $, in  $G-X$ induced by the vertices $\{ v, v', \prv(v), v^*, \glb, \nxt(v)  \}$, which is a \emph{small} obstruction  containing $v$, a contradiction.	 Else, when $ (v', v^*) \in E(G) $, we	 can a find a $ \overline{C_3^*} $, in  $G-X$ induced by the vertices $\{ v', v, v^*, \glb  \}$, which is a \emph{small} obstruction  containing $v$  in $G-X$, a contradiction.	
		
		% 		\textbf{If} there exists a vertex $ v' \in S$ such that  $ v' $ is non adjacent to each of $\{\glb, \prv(v),  \nxt(v)\}$, then we get an induced subgraph $ \overline{C_3^*} $, in  $G-X$ induced by the vertices $\{  v, v',  \prv(v), \nxt(v)\}$, which is a small obstruction in $G-X$. This is  a contradiction to the fact that there is no small obstruction containing $ v  $ in $ G -X$.
		
		\textbf{If} there exists a vertex $ v' \in S \setminus X$ such that $ (v',\glb) \notin E(G) $ and $ v' $ is adjacent to both $  \prv(v)$ and $ \nxt(v)$, then using the same arguments as in case A.1.~, we can find an induced subgraph $\obs'$ in $ G $ which  is an obstruction in $(G-v)-X$, a contradiction.

		\textbf{Else}, each vertex in $ S $ that is non-adjacent to $ \glb $, must be adjacent to exactly one of  $ \prv(v) $ and  $ \nxt(v) $. During the procedure \texttt{Mark-1}, in $ S $, we would have added the left most $ (k+1) $ vertices, say $ S_1 $ and the right most $ (k+1) $ vertices, say $ S_2 $ from $ \cl$. Since $ |S_1| > k $ and $ |S_2|>k $, there  exist some vertices in each of $ S_1 $ and $ S_2 $, that are not in $ X $. Let $ v'_1 $ and $ v'_2 $ be a pair of such vertices  (arbitrarily chosen) from $ S_1 \setminus X $ and $ S_2 \setminus X $, respectively.  Without loss of generality, let $ (v'_1,\prv(v) )  \in E(G)$ and $ (v'_2,\nxt(v) )  \in E(G)$. 
		% Clearly, $ (v'_2,\prv(v) )  \in E(G)$ and $ (v'_1,\nxt(v) )  \in E(G)$.
		Since $ u$ is non-adjacent to both $ v'_1$ and $v'_2$, there exists a pair of vertices $ u_1 $ and $ u_2 $ 	in $ \obs $ such that $ (u_1,v'_1) \in E(G) $, $ (u_2,v'_2) \in E(G) $ with an induced path  $ P $ in the obstruction $\obs $ between $ u_1 $ and $ u_2 $ passing through $ u $ where $ N(\{v'_1, v'_2 \}) \cap V(P -u_1 -u_2) = \emptyset$. Clearly, $ v \notin P $, as $ (v, v'_1), (v, v'_2) \in E(G)$.  Let $ C $ be the cycle induced by the vertices $ V(P) \cup \{v'_1, v'_2\} $. Since $ P \cap X =\emptyset$, $ v'_1, v'_2 \notin X $, $ N(\glb) \cap \{v'_1, v'_2\} = \emptyset $, $\obs':= C \cup \{ \glb\}$ is an obstruction in $(G-v)-X$, a contradiction.
	\end{enumerate}

	\begin{description}
		\item[Case B]  Here we consider the case when $ v $ is the \centre in the \monad $ \obs:= C^*_{\ell}$, i.e., $ v= \glb $. So there exists a cycle $H$ in $ \obs $ such that $\glb$ is not adjacent to any vertex of $V(H)$. We have deleted the vertex $v$ since it was an irrelevant (unmarked)  vertex. Let $ u $ be a vertex in $ \obs \cap T $. Clearly, $ (u,v) \notin E(G) $. Let $ \cl $ be a clique in $ \mathcal{Q} $ containing $ v $. During  Procedure \markone, we have added a set $ S: =\markk_i[A,B]$ of at least $2(k+1)$ vertices from $ \cl $ where $ A= \emptyset $ and $ B= \{u\} $. Otherwise, we would have added $v$  to $ T(\cl)$.  Since $ |S| > k $, there  exists some vertex in $ S $ which is not in  $ X $. Let $ v' $ be such a vertex (arbitrarily chosen) from $ S \setminus X $.
		\begin{itemize}
			\item \textbf{If} the vertex $ v' $ has no neighbour in $ H $, then $\obs':= H \cup \{ v'\}$ is an obstruction in $(G-v)-X$, a contradiction.
			
			\item \textbf{If} the vertex $ v' $ has  exactly one neighbour, say $ u' $ in $ H $, then  we can a find a $ \overline{C_3^*} $, in  $(G-v)-X$, induced by the vertices $\{ u',  \prv(u'), \nxt(u'), v' \}$, which is an obstruction in $(G-v)-X$, a contradiction. 
			
			\item \textbf{If} the vertex $ v' $ has  exactly two neighbours, say $ u' $ and $ u'' $ in $ H $, then  we argue as following. \textbf{When} $ (u', u'') \notin E(G) $, we can  find a $ \overline{C_3^*} $, in  $G-X$ induced by the vertices $\{ v',  u', u'', \glb \}$, a contradiction.			\textbf{When} $  (u' , u'') \in E(G)  $, then 	we can find a $ \overline{S_3} $, in  $G-X$ induced by the vertices $\{ v', u', u'',   \prv(u'), \nxt(u''), \glb \}$. For both these sub-cases, we are able to find \emph{small} obstructions in $G-X$ containing $v$, which is  a contradiction. 
			\item \textbf{Else} the vertex $ v' $ has  at least three neighbours say $ u_1, u_2 $ and $ u_3 $ in $ H $ where $ (u_1, u_3) \notin E(G) $ (note that such $ u_1, u_3 $ always exist as $H$ is large). Then   we can  find a $ \overline{C_3^*} $, in  $G-X$ induced by the vertices $\{ v',  u_1, u_3, \glb\}$, which is a \emph{small} obstruction in $G-X$ containing $v$, a contradiction. 
		\end{itemize} 
		%\textbf{If} the neighbors of $ v' $ in $ H $ is not continuous then there must exists a pair of vertices $ x , y \in H$ and a path $ P \subseteq H $, between them, of length at least two, such that $ x, y \in N(v') $ and $ N(v') \cap V(P -x -y) = \emptyset$. Let $ C $ be the cycle induced by the vertices $ V(P) \cup \{v'\} $. Now  $ P \cap X =\emptyset$, $ v' \notin X $, $ (v', u) \notin E(G) $ together conclude that $\obs':= C \cup \{ u\}$ is an obstruction in $(G-v)-X$, which contradicts that $X$ is a solution to $G-v$ of \phcad. From now onwards we assume that 	the neighbors of $ v' $ in $ H $ is continuous.	 		
\end{description}This completes the proof. \qed \end{proof}

\begin{tcolorbox}[colback=gray!5!white,colframe=gray!75!black]
	\vspace{-1mm}
	\noindent{\bf Procedure \markone.} Let $ \cl $ be a clique.   For  a pair of disjoint subsets $ A, B \subseteq T $, where $ |A| \leq 2 $ and $ |B| \leq 2 $, let  $\markk_i[A,B]$ be the set defined by  $\{ v \in \cl ~|~ A \subseteq N(v), ~B \cap N(v) = \emptyset\}$. We
initialize $ T(\cl)= \emptyset $, and do as follows: 

\begin{itemize}
	\item If $|\markk_i[A,B]|\leq 2(k+1)$, we add all vertices from the set $\markk_i[A,B]$ to $ T(\cl) $.
	
	\item Else, we add the left most $ (k+1) $ vertices (clockwise order of vertices according to their corresponding arc representation) and the right most $ (k+1) $ vertices (anticlockwise order) in $ \markk_i[A,B]$ to $ T(\cl) $.
	
\end{itemize}
\end{tcolorbox}

With the help of 	reduction rule \ref{rule:bound_clique}, after deleting all unmarked vertices from each $\cl \in \mathcal{Q} $, size of each clique $ \cl $ is reduced to $ k^{\OO(1)} $. Therefore, we have the following result. Notice that $\mathcal{Q}$ (with the reduced cliques) is also a \emph{nice clique partition} of $G'-T$ in the reduced instance $(G',k)$.

\begin{lemma}
	\label{lem-clique_size}
	Given an instance $(G, k)$ of \phcad and a nice modulator $T \subseteq V(G)$ of size $k^{\OO(1)}$, in polynomial time, we can construct an equivalent instance $(G', k)$ such that $T \subseteq V(G')$ and there exists a nice clique partition $\mathcal{Q}$ of $G'-T$ such that the size of each clique in $\mathcal{Q}$ is bounded by $ k^{\OO(1)}$. 
\end{lemma}

\section{Bounding the Size of each Connected Component}\label{sec-size_component}

%\section{new format starts from here ->}
From \cref{lem-clique_size}, we can assume that the size of every clique in the \emph{nice clique partition} $\cd=(Q_1,\ldots)$ of $G-T$ for a given instance $(G,k)$ is bounded by $k^{\cO(1)}$.  In this section, we will bound the size of each connected component in $G-T$. For this purpose, it is sufficient to bound the number of cliques $Q_i$'s from $\cd$ appearing in each connected component. 

Let $\mathcal{C}$ be such a connected component. Without loss of generality, we assume that $\mathcal{C}=\bigcup_{i}(Q_i)$ i.e. in the \emph{nice clique partition} $\cd$, in the connected component $\mathcal{C}$, the cliques appear in clockwise direction starting from $Q_1$ as $Q_1, Q_2, \ldots$ etc. We denote $(Q_1,Q_2,\ldots)$  from $\mathcal{C}$ by $\cdc$.

\begin{reduction rule}\label{rr1}
Let $v$ be a vertex in $T$. If $v$ is contained in at least $k+1$ distinct claws $(v,a_i,b_i,c_i)$ intersecting exactly at $\{v\}$, where $a_i,b_i,c_i\in V(G)\setminus T$ then delete $v$ from $G$, and reduce $k$ by $1$. The resultant instance is $(G-v,k-1)$.
\end{reduction rule}

The correctness of the above reduction rule is easy to see as every solution to $(G,k)$ of \phcad must contain the vertex $v$. From here onward we assume that the reduction rule  \ref{rr1} is no longer applicable.

\begin{reduction rule}\label{rr2}
Let $v$ be a vertex in $T$. If $v$ has neighbors in more than $6(k+1)$ different $Q_i$'s ($a_i$'s being the corresponding neighbors), then remove $v$ from $G$ and reduce $k$ by $1$. The resultant instance is $(G-v,k-1)$.
\end{reduction rule}

\begin{lemma}\label{red:Boundnumber}
Reduction Rule \ref{rr2} is safe.
\end{lemma}

\begin{proof}
	By the pigeonhole principle, there are at least $3(k+1)$ non-consecutive cliques that have neighbors of $v$. Let these non-consecutive cliques be denoted by $(Q_1',Q_2',\ldots)$. Now we can construct a set of $k+1$ mutually distinct claws formed by $\{v,a_i,a_{i+1},a_{i+2}\}$ intersecting exactly and only at $\{v\}$ where each $a_j\in Q_j'$. But this implies that any solution to $(G,k)$ of \phcad must contain the vertex $v$. \qed
\end{proof}

 From now on, we assume that the reduction rules \ref{rr1} and \ref{rr2} are no longer applicable i.e.~every vertex $v\in T$ has neighbors in at most $6(k+1)$ different $Q_i$' from $\cdc$. And we have the following result. 

\begin{lemma} \label{lem:seventkplus1}
Let $\cal C$ be a connected component in $G-T$. Then there are at most $6(k+1)|T|$ many distinct cliques $Q_i$'s from $\cdc$ such that $N(T)\cap Q_i \neq \emptyset$. 
\end{lemma}

If  $\cdc$ has more than $300|T|k(k+1)$ cliques, then by the pigeonhole principle and  \cref{lem:seventkplus1}, there are at least $50k$ consecutive cliques  that do not contain any vertex from $N(T)$. Let $Q_1, Q_2, \ldots, Q_{50k}$ 
be the set of $50k$ such consecutive cliques in $\cdc$ which are disjoint from $N(T)$.  
Let $\mathcal{D}_L=\{Q_i \mid i\in[15k,20k]\}$ , $\mathcal{D}_R=\{Q_i \mid i\in[30k,35k]\}$, $F=\{Q_i \mid i\in[20k+1,30k-1]\}$
and $Z=\mathcal{D}_L\cup \mathcal{D}_R\cup F$. Observe that, 
for a vertex $v\in Z$ and a vertex $u\in T$, $\dist_G(u,v) \geq 15k$. And hence there can not be any \emph{small} obstruction containing vertices from $Z$  (\cref{ob1}) which we will use to our advantage in many proofs throughout the current section. Let $\tau$ be the size of minimum $(Q_{20k}-Q_{30k})$  cut in $\cdc$.

\begin{reduction rule}\label{rr3}
Let $F$ be as defined above. Delete all the vertices of $F$ from $G$. Introduce a new clique $S$ of size $\tau$. Also, add edges such that $G[V(Q_{20k})\cup S]$  and $G[V(Q_{30k})\cup S]$ are complete graphs. The cliques appear in the order $Q_{20k},S,Q_{30k}$.
\end{reduction rule}

\begin{figure}[t!]
	\begin{center}
		\includegraphics[width=.7\textwidth]{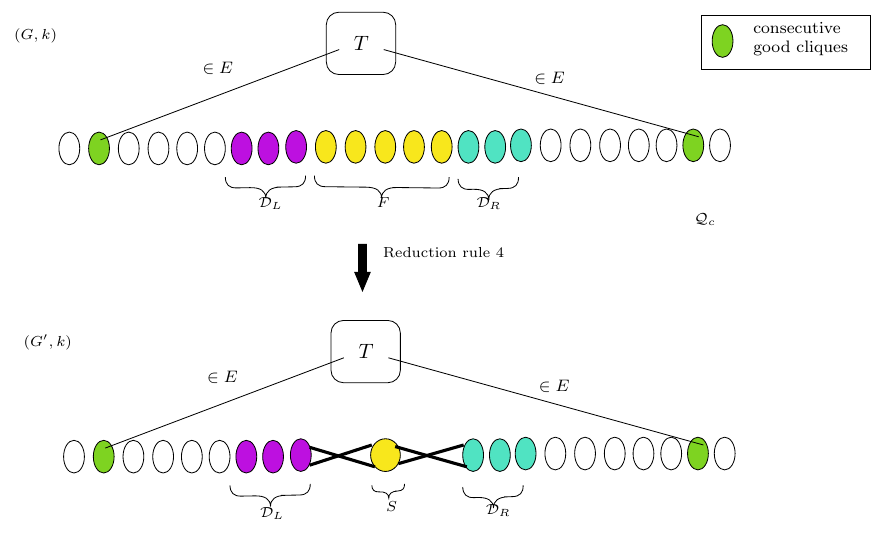}
	\end{center}
	\caption{Description of reduction rule \ref{rr3}.}
	\label{fig:rule5}
\end{figure}

Let $G'$ be the reduced graph after application of the reduction rule \ref{rr3}. For an illustration, see \cref{fig:rule5}. Notice that $G'-T$ is a \phcag~by construction.

\begin{observation}\label{ob1}
There are no \emph{small} obstructions containing any vertices from $ \mathcal{D}_L \cup F \cup \mathcal{D}_R$ in $G$. 
Similarly, there are no \emph{small} obstructions containing vertices of $ \mathcal{D}_L \cup S \cup  \mathcal{D}_R$ in $G'$.
\end{observation}

\begin{proof}
	Let $\obs$ be a \emph{small} obstruction in $G$ such that $V(\obs)\cap  (\mathcal{D}_L \cup F \cup \mathcal{D}_R)\neq \emptyset$.
	Since for any vertex $v\in \mathcal{D}_L \cup F \cup \mathcal{D}_R$ and a vertex $u\in T$, $\dist_G(u,v) \geq 15k$ and $|\obs|\leq \cyclebound$, hence
	$V(\obs)\cap T=\emptyset$. But this is a contradiction, since $G- T$ has no obstructions. So 
	there are no \emph{small} obstructions containing any vertices from $ \mathcal{D}_L \cup F \cup \mathcal{D}_R$.
	
	Let $\obs'$ be a \emph{small} obstruction in $G'$ such that $V(\obs')\cap  (\mathcal{D}_L \cup S \cup \mathcal{D}_R)\neq \emptyset$,
	Since for any vertex $v\in \mathcal{D}_L \cup S \cup \mathcal{D}_R$ and a vertex $u\in T$, $\dist_{G'}(u,v) \geq 15k$ and $|\obs'|\leq \cyclebound$, hence
	$V(\obs')\cap T=\emptyset$. But this is a contradiction, since $G'- T$ has no obstructions. So
	there are no \emph{small} obstructions containing any vertices from $ \mathcal{D}_L \cup S \cup \mathcal{D}_R$.\qed
\end{proof}

 \begin{observation}\label{obs:inGintZ} Any \hole $H$ of a \monad with a \centre $v$  in $G$ which contains a vertex from $\mathcal{D}_L \cup F \cup \mathcal{D}_R$, intersects all cliques in $\mathcal{D}_L \cup F \cup \mathcal{D}_R$. And such an $H$ has size at least $20k$. 
\end{observation}

 \begin{proof}
	Without loss of generality suppose $H$ intersects $Q_{i}$ but does not intersect some $Q_{i+1}\in Z$. Then the clique $Q_{i'}$ where $i'=i-1$ contains at least two vertices from $H$. This is only possible when $H$ has size at most four (since it is a \monad). But then $H$ along with any vertex from $Q_{i-3}$ or $Q_{i+2}$ will form a \emph{small} obstruction completely contained in $G-T$, which is not possible.  Hence $H$ intersects all cliques in $\mathcal{D}_L \cup F \cup \mathcal{D}_R$.\qed
\end{proof}

\begin{observation}\label{obs:inGintZ1}
Any \hole $H$ of a \monad with a \centre $v$ in $G'$ which contains a vertex from $\mathcal{D}_L \cup S \cup \mathcal{D}_R$, intersects all cliques in $\mathcal{D}_L \cup S \cup \mathcal{D}_R$. And such an $H$ has size at least $20k$. 
\end{observation}
\begin{proof}
 Proof is similar to the proof for  \cref{obs:inGintZ}.\qed
\end{proof}

\begin{lemma}\label{red:BoundIngComp}
 Reduction Rule \ref{rr3} is safe.
\end{lemma}

\begin{proof}
	We show that $(G,k)$ is a \yes-instance of \phcad if and only if $(G',k)$ is a \yes-instance of \phcad.
	
	\noindent
	$ (\Rightarrow) $ Suppose $(G,k)$ is a \yes-instance and let $X$ be a minimum size solution. Recall that ${\cal D}_L$ and ${\cal D}_R$ contain $(5k+1)$ cliques each and $\tau$ is the size of minimum $(Q_{20k}-Q_{30k})$ cut. We note that this cut may include the vertices from $Q_{20k}$ and $Q_{30k}$. Let $W$ be the set of vertices from all the cliques in $Q_{20k+1}\cup \ldots \cup Q_{30k-1}$.

	\begin{claim}
		Either $X\cap W=\emptyset$ or $|X\cap W|=\tau$.
	\end{claim}
	\begin{proof}
		Suppose that $X\cap W\neq\emptyset$ and $|X\cap W|< \tau$. Let $v\in X\cap W$. As $X$ is also a minimal solution, for every vertex $u\in X$, there exists an obstruction that does not contain any other vertex from $X$. This implies that there exists a \monad  containing $v$ and not containing any vertex from $X\setminus \{v\}$. We first show that $v$ can not be  \centre  of such a \monad. If it is a \centre of \monad $H\cup \{v\}$, then from \cref{obs:inGintZ1}, $H\cap Z=\emptyset$. But in $Z$ we have $k+1$ many cliques and each vertex  of them can create an obstruction (\monad) with $H$ as an \mhole. Hence $X\cap H$ can not be empty, which is a contradiction to the fact that $X$ was a minimal solution. So the obstruction for $v$ must contain $v$ as one of its \hole vertices i.e. $v\in V(H)$. Let the \centre of this obstruction be $z$. Notice that $X\cap H=\{v\}$. In the \hole $H$, let $v_1$ and $v_2$ be two vertices from $Q_{20k}$ and $Q_{30k}$ respectively with no other vertex from $Q_{20k}\cup Q_{30k}$ in between. Between $v_1$ and $v_2$, we have $\tau$ many vertex disjoint (induced) paths, say $\cal P$ in $W$. Since  $|X \cap W|<\tau$, it does not intersect at least one of the paths from $\cal P$. But then replacing the segment of $H$ between $v_1$ and $v_2$ with the non-intersected path, we get a new \hole $H'$ where $X\cap H'=\emptyset$. Adding $z$ as \centre, we get an obstruction $H'\cup{z}$ which is contained in $G-X$, a contradiction. Hence $|X\cap W|\geq \tau$. If $|X\cap W| > \tau$, we construct a new set $X'=(X\setminus W)\cup X''$, where $X''$ is a $\tau$ sized min-cut between $Q_{20k}$ and $Q_{30k}$, We claim that $G-X'$ is also a \phcag. If not, then there is an obstruction in $G-X'$. But such an obstruction must necessarily contain a vertex from $X\cap W$. Any obstruction containing a vertex $v\in X\cap W$ is a \emph{large} obstruction. If this obstruction contains $v$ as \centre, then by arguments similar to the ones made just above, we can say that $X\setminus W $ also intersects this obstruction, and hence so does $X'$. On the other hand, if the obstruction contains $v$ as a vertex in its \hole, then such an obstruction is hit by the min-cut between $Q_{20k}$ and $Q_{30k}$ in $X'$. But $X'$ is a strictly smaller solution than $X$, which is a contradiction.
		Hence the claim is proved.
	\end{proof}

	Using the above claim we consider the following cases: Recall that $X$ is a solution to the \yes-instance $(G,k)$.
	\smallskip 
	
	\noindent \textbf{Case 1: $X\cap W=\emptyset$}\\ 
	Here we claim that $X$ is also a solution to $(G',k)$. Suppose it is not true. Then there is an obstruction $\obs$  in $G'-X$.  Now $V(\obs)\cap S\neq \emptyset$, otherwise we will have the same obstruction in $G-X$. Hence  this obstruction must be a \emph{large} obstruction (from arguments similar to the ones made in \cref{obs:inGintZ}). Let $v\in (\obs\cap S)$. If $v$ is a \centre in $\obs$, then $H$ ($H$ is the \hole of $\obs$) is contained in $G-(Z\cup X)$. But any vertex from any clique between $Q_{24k}$ and  $Q_{26k}$ with $H$ will form an obstruction in $G-X$. Notice that all these (at least) $2k$ vertices, each form an obstruction with $H$. Hence $X$ must intersect $H$, which is a contradiction to the fact that $H\cup \{v\}$ does not contain any vertex from $X$. For the other case, when $v$ is part of the \hole in $\obs$, let $v_1$ and $v_2$ be two vertices of $H$ from $Q_{20k}$ and $Q_{30k}$ respectively with no other vertex from $Q_{20k}\cup Q_{30k}$ in between. Since $X$ does not intersect any vertex from $W$, replacing the segment of $H\cap S$ with a path between $v_1$ and $v_2$, we get a new \hole $H'$, where $X\cap H'=\emptyset$. Adding $z$ as a \centre we get an obstruction $H'\cup{z}$ which is contained in $G-X$, a contradiction. 
	
	\smallskip 
	
	\noindent \textbf{Case 2: $|X\cap W|=\tau$}\\   
	Here we claim that $X'=(X\setminus W)\cup S$ is also a solution to $(G',k)$. Suppose it is not true. Then there is an obstruction $\obs$ in $G'-X'$. But then  $V(\obs)\cap W\neq \emptyset$. And this obstruction must be a \emph{large} obstruction  (from arguments similar to the ones made in \cref{obs:inGintZ}. Let $v \in (\obs\cap W)$. The vertex $v$ must be a \centre in $\obs$, since $S$ (and hence $X'$) intersects all paths between $Q_{20k}$ and $Q_{30k}$. But then $H$ ($H$ is the \hole of $\obs$) is contained in $G-(Z\cup X)$. And there are at least $k+1$ vertices in $Z$ each of which an obstruction with $H$ and at least one of them, say $u$ is not contained in $X$. Hence $u$ along with $H$ forms a \monad that is contained  in $G-X$, a contradiction. Hence $X'$ is also a solution to $(G',k)$.  This completes the proof in the forward direction.
	
	\medskip 
	
	\noindent
	$ (\Leftarrow) $  Suppose $(G', k)$ is a \yes-instance of \phcad where $Y$ is a minimum  size solution. Let $Z'=\mathcal{D}_L \cup S \cup \mathcal{D}_R$.
	\begin{claim}
		Either $Y\cap S=\emptyset$ or $|Y\cap S|=\tau$.
	\end{claim}
	
	\begin{proof}
		Suppose that $Y\cap S\neq\emptyset$ and $|Y\cap S|< \tau$. Let $v\in Y\cap S$. As $Y$ is a minimal solution, for every vertex $u\in Y$, there exists an obstruction that does not contain any vertex from $Y\setminus \{u\}$. This implies that there exists a \monad  containing $v$ and not containing any vertex in $Y\setminus \{v\}$. We first show that $v$ can not be a \centre of such an obstruction. If it is a \centre of an obstruction induced by $H\cup \{v\}$, then by \cref{ob1} $H\cap Z'=\emptyset$. But then in $Z'$ we have $k+1$ many distinct cliques and each vertex from these cliques can form obstruction with $H$ as a \centre. Hence $Y\cap H$ can not be empty, which is a contradiction to the fact that $Y$ was a minimal solution. So the obstruction for $v$ must contain $v$ as one of its \hole vertices i.e. $v\in H$. Let the \centre of this obstruction be $z$. Notice that $Y\cap H=\{v\}$. In the \hole $H$, let $v_1$ and $v_2$ be two vertices from $Q_{20k}$ and $Q_{30k}$ respectively with no other vertex from $Q_{20k}\cup Q_{30k}$ in between. Between $v_1$ and $v_2$ we have $\tau$ many vertex disjoint (induced) paths  ($\cal P$) (each path consists of exactly one vertex from $S$). $Y$ does not intersect at least one of these paths from $\cal P$. But then replacing the segment of $H$ between $v_1$ and $v_2$ with the non-intersected path from $\cal P$, we get a new \hole $H'$ where $Y\cap H'=\emptyset$. Adding $z$ as the \centre, we get an obstruction $H'\cup{z}$ which is contained in $G'-Y$, a contradiction. Hence $Y\cap S= \tau$. Hence the claim is proved.
	\end{proof}
	
	Using the above claim we consider the following cases: Recall that $Y$ is a minimum size solution to the \yes-instance  $(G',k)$.
	
	\smallskip
	
	\noindent \textbf{Case 1: $Y\cap S=\emptyset$}\\ 
	Here we claim that $Y$ is also a solution to $(G,k)$. Suppose it is not true. Then there is an obstruction $\obs$ in $G-Y$.  $\obs\cap W\neq \emptyset$, otherwise we will have the same obstruction in $G'-Y$. Hence  this obstruction must be a \emph{large} obstruction (from arguments similar to the ones made in \cref{obs:inGintZ1}). Let $v\in (\obs\cap W)$. If $v$ is a \centre in $\obs$, then $H$ ($H$ is the \hole of $\obs$) is contained in $G-(Z'\cup Y)$. But there are at least $k+1$ vertices from $Z'$ who along with $H$ form obstructions in $G'-Y$. Hence $Y$ must intersect $H$, which contradicts the fact that $H\cup\{v\}$ is an obstruction in $G-Y$ . For the other case when $v$ is part of the \hole in $\obs$, let $v_1$ and $v_2$ be two vertices from $Q_{20k}$ and $Q_{30k}$ respectively with no other vertex from $Q_{20k}\cup Q_{30k}$ in between in $H$. Since $Y$ does not contain any vertex from $S$, by replacing the segment of $H$ between $v_1$ and $v_2$ with any vertex of $S$, we get a new \hole $H'$ where $Y\cap H'=\emptyset$. Adding $z$ as the \centre we get an obstruction $H'\cup{z}$ which is contained in $G'-Y$, which is a contradiction.
	
	\smallskip  
	
	\noindent \textbf{Case 2: $|Y\cap S|=\tau$}\\   
	Here we claim that $Y'=(Y\setminus S)\cup Y''$, where $Y''$ is a $\tau$ sized min-cut between $Q_{20k}$ and $Q_{30k}$, is a solution to $(G,k)$. Suppose it is not true. Then there is an obstruction $\obs$ in $G-Y'$. But  $\obs\cap W\neq \emptyset$ and $\obs$ must be a \emph{large} obstruction (from arguments similar to the ones made in \cref{obs:inGintZ1}). Let $v \in (\obs\cap W)$. Since there is no path from $Q_{20k}$ to $Q_{30k}$ in $G-Y'$, any obstruction must contain $v$ as its \centre only. But then $H$ ($H$ is the \hole of $\obs$)  is contained in $G-(Z\cup Y')$ and hence also in $G'-(Z'\cup Y)$. And there are $k+1$ many vertices in $Z'$ who can form obstructions with the \hole $H$. At least one of them is not contained in $Y$. This vertex along with $H$ forms an obstruction that is contained  in $G'-Y$, a contradiction. Hence $Y$ is also a solution to $(G,k)$.
	This completes the proof in the reverse direction.
	
	\noindent This concludes the proof for the lemma.\qed
\end{proof}

 With reduction rule \ref{rr3}, we obtain the following result. 
\begin{lemma}
\label{lem:final-connectedsize}
Given an instance $(G, k)$ of \phcad and a nice modulator  $T \subseteq V(G)$ of size $k^{\OO(1)}$, in polynomial time, we can construct an equivalent instance $(G', k)$ such that, $T \subseteq V(G')$ is a nice modulator for $G'$ and for each connected component $\C{C}$ of $G' - T$, the number of cliques in $\cdc$  is at most $300\cdot |T|\cdot k(k+1) = \OO(k^2 \cdot \card{T})$. 
\end{lemma}

\section{Bounding the Number of Connected Components }\label{sec-number_component}
% \begin{lemma} \label{lem-number_components}
%  Given an instance $(G, k)$ and a nice modulator $T \subseteq V (G)$ of size $\OO(k^{12})$, in polynomial-time, we can construct an equivalent instance $(G'
% , k')$ such that the number of connected component in $G' - T$ is $\OO(k \cdot |T|^2)$.
% \end{lemma}

Until now we have assumed that $G - T$ is connected. Further, in \cref{sec-size_component}, we showed that the size of any connected component is upper bounded by $k^{\OO(1)}$. In this section, we show that the number of connected components in $G - T$ can be upper bounded by $k^{\OO(1)}$. This together with the fact that $|T| \leq k^{\OO(1)}$, results in a polynomial kernel for \phcad.

Here we  bound the number of connected components with an argument similar to the one using which we bounded the neighborhood of the modulator. We make use of the claw obstruction to get the desired bound. Notice that if any vertex $v$ in $T$ has neighbors in three  different components in $G - T$, then we get a claw. 

\begin{reduction rule} \label{rule 3k}
Let $v$ be a vertex in $T$ such that $v$ has neighbors in at least $3(k+1)$  different components in $G - T$ then delete
$v$ from $G$, and reduce $k$ by $1$. The resultant instance is $(G - v, k - 1)$.
\end{reduction rule}

The correctness of the above reduction rule is easy to see as every solution to $(G, k)$ of \phcad~must contain $v$. From now onwards we assume that reduction rule \ref{rule 3k} is not applicable. And this leads to the following lemma.

\begin{lemma}\label{lem-nbr_components}
$T$ can have neighbors in at most $3(k+1)|T|$ many different components. 
\end{lemma}

%A connected component that has no neighbor in $T$ is a \phcag. Hence, we can safely remove this component from our given instance.

Now we bound the number of connected components that have no neighbor in $T$. Towards that, we classify all such connected components into two classes: interval connected components (which admit an interval representation) and non-interval connected components. Here non-interval connected components cover the entire circle whereas others partially cover the underlying circle.

\begin{clm}\label{claim:mostone}
The number of non-interval connected components in $G-T$ is at most one. 
\end{clm}

\begin{proof}
	For contradiction suppose there are at least two non-interval connected components $C_1, C_2$ in $G-T$. Notice  that in any proper Helly circular-arc representation $\sigma$ of $G-T$, the circular-arcs corresponding to all the vertices in $C_1$ together cover the entire circle in the representation. But then there is no other connected component in $G-T$ that can admit arc representation $\sigma$. Hence $C_2$ can not have any circular-arc representation in $\sigma$, a contradiction.\qed
\end{proof}

%\begin{claim}
%The number of interval connected components in $G-T$ is at most $k+1$. 
%\end{claim}

\begin{reduction rule}\label{rule no nbr}
If there are more than $(k+1)$ interval connected components in $G - T$ that have no neighbor in $T$, delete all but $(k+1)$ components. 
\end{reduction rule}

\begin{lemma}\label{lem:rule7}
	Reduction Rule \ref{rule no nbr} is safe.
		\end{lemma}

\begin{proof}
	Let $(G',k)$ be the reduced instance. We show that $(G,k)$ is a \yes-instance of \phcad if and only if $(G',k)$ is a \yes-instance. The forward direction is trivial as $G'$ is an induced subgraph of $G$. In the backward direction, let $(G',k)$ be a \yes-instance. Assume that $X$ is a solution for \phcad on $(G', k)$ and $\sigma$ is a proper Helly circular-arc representation of $G'-X$.   As there are $(k+1)$ interval connected components in $G-T$ there always exists a component $C$ in $G-T$ which has no intersection with $X$ and hence all the vertices in $C$ belong to $G'-X$. Now we can always shrink all the arcs corresponding to $C$ in $\sigma$ to half of their length. In the freed-up space, we can accommodate the arcs corresponding to the deleted interval connected components. Hence $(G,k)$ is a \yes-instance.\qed
\end{proof}

From now onwards we assume that reduction rules \ref{rule 3k} and \ref{rule no nbr} are not applicable. Now these two reduction rules  and
\cref{lem-nbr_components} implies the following result:

\begin{lemma} \label{lem-number_components}
 Given an instance $(G, k)$ and a nice modulator $T \subseteq V (G)$ of size $\OO(k^{12})$, in polynomial-time, we can construct an equivalent instance $(G'
, k')$ such that the number of connected component in $G' - T$ is $\OO(k \cdot |T|^2)$.
\end{lemma}
% Due to limited available space, we have relocated the proof associated with Lemma\ref{lem-number_components} to \cref{sec-number_component}.
\section{Kernel size analysis}

Now we are ready to prove the main result of our paper, that is, \cref{theo:poly_kernel}.
Before proceeding with the proof, let us state all the bounds that contributes to the kernel size.
\begin{tcolorbox}
% Size of $T'$:
%\\Size of $T''$:\\
Size of \emph{nice modulator} $T$: $\OO(k^{12})$ (from \cref{nicemod}).
\\Number of connected components in $G-T$: $\mathcal{O}(k\cdot{\card{T}}^2)$ (by \cref{lem-number_components}).
\\Number of cliques in any connected component in $G-T$: $\OO(k^2\cdot |T|)$ (by \cref{lem:final-connectedsize}).
% \\Size of any $Q_i$ in $G-T$: 
% \\Size of any $R_i$ in $G-T$:
\\Size of any clique $\cl$ in $G-T$:  $2(k+1) \card{T}^4$ (by \cref{rem:Mark1}). 
\end{tcolorbox}

%\begin{theorem}
% \bpvdfull admits a polynomial kernel.
%\end{theorem}

\begin{proof}[Proof of \cref{theo:poly_kernel}]
From \cref{lem-efficient} and \cref{lem:redundant}, in polynomial-time, we can obtain a \emph{nice modulator} $T\subseteq V(G)$ of size $\mathcal{O}(k^{\cyclebound})$ or conclude that $(G,k)$ is a \no-instance. Note that, $G-T$ is a \phcag. Next, we take a \emph{nice clique partition} of $G-T$. Now by \cref{lem-number_components}, in polynomial-time we return a graph $G$ such that $G-T$ has  $\OO(k\cdot |T|^2)$ components.  
By Lemma \cref{lem:final-connectedsize}, in polynomial-time, we can reduce the graph $G$ such that any connected component in $G-T$ has at most  $\OO(k^2\cdot |T|)$ cliques. 
Next, we bound the size of each clique in $G-T$ by  $2(k+1) \card{T}^4$ from \cref{lem-clique_size}. Hence the graph $G-T$ has at most $ \OO(k\cdot |T|^2)\cdot  \OO(k^2\cdot |T|) \cdot 2(k+1) \card{T}^4 $, that is, $\OO(k^4\cdot |T|^7)$ many vertices. Recall that $|T|=\OO(k^{12})$. Therefore, the size of the obtained kernel is $\OO(k^4\cdot |T|^7)$, that is, $\OO(k^{88})$. \qed
\end{proof}

\section{Conclusion}

In this paper, we studied \phcad from the perspective
of kernelization complexity, and designed a polynomial kernel of size $\OO(k^{88})$.  We remark that the size of a kernel can be further optimized with more careful case analysis. However, getting a kernel of a significantly smaller size might require an altogether different approach.

	\bibliography{draft.bib}

\begin{thebibliography}{10}

\bibitem{AgrawalLMSZ17}
Akanksha Agrawal, Daniel Lokshtanov, Pranabendu Misra, Saket Saurabh, and
  Meirav Zehavi.
\newblock {Feedback Vertex Set Inspired Kernel for Chordal Vertex Deletion}.
\newblock In {\em {(SODA 2017)}}, pages 1383--1398, 2017.

\bibitem{AgrawalM0Z19}
Akanksha Agrawal, Pranabendu Misra, Saket Saurabh, and Meirav Zehavi.
\newblock Interval vertex deletion admits a polynomial kernel.
\newblock In Timothy~M. Chan, editor, {\em {SODA} 2019, San Diego, California,
  USA}, pages 1711--1730. {SIAM}, 2019.

\bibitem{BDFH09}
Hans~L. Bodlaender, Rodney~G. Downey, Michael~R. Fellows, and Danny Hermelin.
\newblock On problems without polynomial kernels.
\newblock {\em {Journal of Computer and System Sciences}}, 75(8):423--434,
  2009.

\bibitem{DBLP:journals/corr/abs-2202-00854}
Yixin Cao, Hanchun Yuan, and Jianxin Wang.
\newblock {Modification Problems Toward Proper (Helly) Circular-Arc Graphs}.
\newblock In J{\'{e}}r{\^{o}}me Leroux, Sylvain Lombardy, and David Peleg,
  editors, {\em {MFCS} 2023}, volume 272 of {\em LIPIcs}, pages 31:1--31:14.
  Schloss Dagstuhl, 2023.

\bibitem{paramalgoCFKLMPPS}
Marek Cygan, Fedor~V. Fomin, {\L}ukasz Kowalik, Daniel Lokshtanov, Daniel Marx,
  Marcin Pilipczuk, Michal Pilipczuk, and Saket Saurabh.
\newblock {\em {Parameterized Algorithms}}.
\newblock Springer-Verlag, 2015.

\bibitem{diestel-book}
Reinhard Diestel.
\newblock {\em Graph Theory, 4th Edition}, volume 173 of {\em Graduate texts in
  mathematics}.
\newblock Springer, 2012.

\bibitem{DBLP:series/mcs/DowneyF99}
Rodney~G. Downey and Michael~R. Fellows.
\newblock {\em {Parameterized Complexity}}.
\newblock Monographs in Computer Science. Springer, 1999.

\bibitem{DowneyFbook13}
Rodney~G. Downey and Michael~R. Fellows.
\newblock {\em {Fundamentals of Parameterized Complexity}}.
\newblock Texts in Computer Science. Springer, 2013.

\bibitem{duran2014structural}
Guillermo Dur{\'a}n, Luciano~N Grippo, and Mart{\'\i}n~D Safe.
\newblock Structural results on circular-arc graphs and circle graphs: a survey
  and the main open problems.
\newblock {\em Discrete Applied Mathematics}, 164:427--443, 2014.

\bibitem{DBLP:series/txtcs/FlumG06}
J{\"{o}}rg Flum and Martin Grohe.
\newblock {\em {Parameterized Complexity Theory}}.
\newblock Texts in Theoretical Computer Science. An {EATCS} Series. Springer,
  2006.

\bibitem{DBLP:journals/siamdm/FominSV13}
Fedor~V. Fomin, Saket Saurabh, and Yngve Villanger.
\newblock A polynomial kernel for proper interval vertex deletion.
\newblock {\em {SIAM} J. Discret. Math.}, 27(4):1964--1976, 2013.

\bibitem{Golumbic80}
Martin~Charles Golumbic.
\newblock {\em {Algorithmic Graph Theory and Perfect Graphs}}.
\newblock Academic Press, 1980.

\bibitem{kaplan2009certifying}
Haim Kaplan and Yahav Nussbaum.
\newblock Certifying algorithms for recognizing proper circular-arc graphs and
  unit circular-arc graphs.
\newblock {\em Discrete Applied Mathematics}, 157(15):3216--3230, 2009.

\bibitem{KeCOLW18}
Yuping Ke, Yixin Cao, Xiating Ouyang, Wenjun Li, and Jianxin Wang.
\newblock Unit interval vertex deletion: Fewer vertices are relevant.
\newblock {\em J. Comput. Syst. Sci.}, 95:109--121, 2018.

\bibitem{DBLP:conf/latin/0001S0Z18}
R.~Krithika, Abhishek Sahu, Saket Saurabh, and Meirav Zehavi.
\newblock The parameterized complexity of cycle packing: Indifference is not an
  issue.
\newblock {\em Algorithmica}, 81(9):3803--3841, 2019.

\bibitem{DBLP:journals/dam/LinSS13}
Min~Chih Lin, Francisco~J. Soulignac, and Jayme~Luiz Szwarcfiter.
\newblock {Normal Helly circular-arc graphs and its subclasses}.
\newblock {\em Discret. Appl. Math.}, 161(7-8):1037--1059, 2013.

\bibitem{lin2009characterizations}
Min~Chih Lin and Jayme~L Szwarcfiter.
\newblock {Characterizations and recognition of circular-arc graphs and
  subclasses: A survey}.
\newblock {\em Discrete Mathematics}, 309(18):5618--5635, 2009.

\bibitem{DBLP:conf/cocoon/LinS06}
Min~Chih Lin and Jayme~Luiz Szwarcfiter.
\newblock {Characterizations and Linear Time Recognition of Helly Circular-Arc
  Graphs}.
\newblock In Danny~Z. Chen and D.~T. Lee, editors, {\em {COCOON} 2006}, volume
  4112 of {\em LNCS}, pages 73--82. Springer, 2006.

\bibitem{Marx10}
D{\'a}niel Marx.
\newblock {Chordal Deletion is Fixed-Parameter Tractable}.
\newblock {\em Algorithmica}, 57(4):747--768, 2010.

\bibitem{mcconnell2003linear}
Ross~M McConnell.
\newblock Linear-time recognition of circular-arc graphs.
\newblock {\em Algorithmica}, 37(2):93--147, 2003.

\bibitem{DBLP:books/ox/Niedermeier06}
Rolf Niedermeier.
\newblock {\em {Invitation to Fixed-Parameter Algorithms}}.
\newblock Oxford University Press, 2006.

\end{thebibliography}

%\appendix
%\input{Appendix}

\end{document}